\documentclass[11pt]{article}

\usepackage[OT4]{fontenc}
\usepackage[utf8]{inputenc}
\usepackage[top=1in, bottom=1in, inner=1in, outer=1in]{geometry}
\usepackage{amsmath,amsfonts,amssymb,mathrsfs,dsfont,amsthm} \usepackage{mathtools}
\usepackage{graphicx}
\usepackage{comment}
\usepackage{enumerate}
\usepackage{bbm}
\usepackage{float}
\usepackage{tikz}
\usepackage{textcomp}

\usetikzlibrary{arrows.meta}

\newtheorem{theorem}{Theorem}[section]
\newtheorem{corollary}{Corollary}[section]
\newtheorem{remark}{Remark}[section]
\newtheorem{claim}{Claim}[section]
\newtheorem{definition}{Definition}[section]
\newtheorem{proposition}{Proposition}[section]
\newtheorem{lemma}{Lemma}[section]

\newcommand{\R}{\mathbb{R}}
\newcommand{\eps}{\varepsilon}
\newcommand{\vecspan}{\mathrm{span}}

\newcommand{\uhat}{\widehat{u}}

\newcommand{\cafter}{\widetilde{c}}
\newcommand{\uafter}{\widetilde{u}}
\newcommand{\ctwo}{c_{\mathrm{two}}}
\newcommand{\cone}{c_{\mathrm{one}}}
\newcommand{\vone}{v_{\mathrm{one}}}
\newcommand{\vtwo}{v_{\mathrm{two}}}
\newcommand{\vbar}{\overline{v}}
\newcommand{\alphawhat}{\widehat{\alpha}}

\DeclareMathOperator{\sgn}{sgn}
\DeclareMathOperator{\Span}{span}

\DeclareMathOperator*{\EE}{E}

\title{A Geometric Model of Opinion Polarization}

\author{Jan Hązła\footnote{Email: \url{jan.hazla@epfl.ch}. Partially supported by DMS-1737944.} 
\and
Yan Jin\footnote{Email: \url{yjin1@mit.edu}. Partially supported by ARO MURI W911NF1910217.}
\and
Elchanan Mossel\footnote{Email: \url{elmos@mit.edu}. Partially supported by 
by Simons Investigator in Mathematics award (622132), NSF DMS award 1737944 
and CCF award 1918421, ARO MURI grant W911NF1910217 and by a Vannevar Bush Faculty Fellowship ONR-N00014-20-1-2826.}
\and
Govind Ramnarayan\footnote{Email: \url{govind@mit.edu}. Partially supported by DMS-1737944 and ARO MURI W911NF1910217.}}

\date{}

\newcounter{eqn}
\newcommand{\num}{\refstepcounter{eqn}\text{\theeqn}\;}

\makeatletter
\newcommand{\putindeepbox}[2][0.7\baselineskip]{{%
    \setbox0=\hbox{#2}%
    \setbox0=\vbox{\noindent\hsize=\wd0\unhbox0}
    \@tempdima=\dp0
    \advance\@tempdima by \ht0
    \advance\@tempdima by -#1\relax
    \dp0=\@tempdima
    \ht0=#1\relax
    \box0
}}
\makeatother

\usepackage[colorlinks=true, allcolors=blue]{hyperref}

\begin{document}

\maketitle

\begin{abstract}
  We introduce a simple, geometric model of opinion polarization.
  It is a model of political persuasion,
  as well as marketing and advertising, utilizing social values.
  It focuses on the interplay between different topics and 
  persuasion efforts. 
  We demonstrate that societal opinion polarization often arises as an unintended byproduct of influencers attempting to promote a product or idea.
  We discuss a number of mechanisms for the emergence of polarization involving one or more influencers, sending messages strategically, heuristically, or randomly. We also examine some computational aspects of choosing
  the most effective means of influencing agents,
  and the effects of those strategic considerations
  on polarization.
\end{abstract}

\section{Introduction}

Opinion polarization is a widely acknowledged social phenomenon, especially
in the context of political opinions \cite{FA08, SH15, IW15}, leading to
recent concerns over ``echo chambers'' created by mass media \cite{Pri13}
and social networks \cite{Con11, Pariser, BMA15, BAB18, Gar18}. The objective of this paper is
to propose a simple, multi-dimensional geometric model of the dynamics of polarization where the 
evolution of correlations between opinions on different topics plays a key role. 

Many models have been proposed to explain how polarization arises, and this remains an active area of research
\cite{NSL90, Axe97, Noa98, HK02, MKFB03, BB07, DGL13, DVSC17, KP18, PY18, SCP19}.
Our attempt aims at simplicity over complexity.
As opposed to a large majority of previous works, our model does not require social network-based 
mechanism. 
Instead, we focus
on influences of advertising or political \emph{campaigns} that reach a
wide segment of the population.

We develop a high-dimensional variant of \emph{biased assimilation} \cite{LRL79}
and use it as
our main behavioral assumption. The bias assimilation for one topic states that people tend to be receptive to opinions they agree with, and antagonistic to opinions they disagree with.

The \emph{multi-dimensional} setting
reflects the fact that campaigns often touch on many topics.
For example, in the context of American politics, one might wonder why
there exists a significant correlation between opinions of individuals on,
say, abortion access, gun rights and urgency of climate change
\cite{Pew14}.
Our model attempts to illustrate how such correlations between opinions
can arise as a (possibly unintended) effect of advertising exploiting
different topics and social values.
  
In mathematical terms, we consider a population of agents with preexisting opinions
represented by vectors in $\mathbb{R}^d$.
Each coordinate represents a distinct topic, and the value of the coordinate reflects the agent's opinion on the topic, which can be positive or negative.
As discussed more fully in Section~\ref{sec:design-choices}, we assume that all opinions lie on the
Euclidean unit sphere. This reflects an assumption that
each agent has the same ``budget of importance'' of different
topics.
We then consider a sequence of
\emph{interventions} affecting the opinions. 
An intervention is also
a unit vector in $\mathbb{R}^d$, representing the set of opinions
expressed in, e.g., an advertising campaign or ``news cycle''.

We model the effect of intervention $v$ on an agent's opinion $u$ in the following way. Supposing an agent starts with opinion $u \in \mathbb{R}^d$, after receiving an intervention $v$ it will update the opinion to the unit vector proportional
to
\begin{align}\label{eq:07}
  w = u + \eta\cdot\langle u,v\rangle\cdot v \;,
\end{align}
where $\eta > 0$ is a global parameter that controls the influence of an
intervention. Most of our results do not depend on a choice of $\eta$ and
in our examples we often take $\eta=1$ for the sake of
simplicity. Smaller values of $\eta$ could model campaigns with limited
persuasive power. This and other design choices
are discussed more extensively in Section~\ref{sec:design-choices}.

Intuitively, the agent evaluates the received message in context of its existing opinion, and assimilates this message weighted by its ``agreement'' with it. 
Our model exhibits biased assimilation in that
if the intervening opinion $v$ is positively correlated with an
agent's opinion $u$, then after the update the agent opinion moves towards
$v$, and conversely, if $v$ is negatively correlated with $u$, then the update
moves $u$ away from $v$ and towards the opposite opinion $-v$.

One way to think of the intervention is as an exposure to persuasion by a political actor, like
a political campaign message. 
A different way, in the context of marketing,
is a product advertisement that
exploits values besides the quality of the product.
In that context, we can think of one of
the $d$ coordinates of the opinion vector as representing
opinion on a product being introduced into the market and the
remaining
coordinates as representing preexisting opinions on other (e.g., social or political)
issues. Then, an intervention would be an advertising effort
to connect the product with a certain set of opinions or values~\cite{VSL77}.
Some examples are corporate advertising campaigns
supporting LGBT rights \cite{Sny15} or gun manufacturers associating
their products with patriotism and conservative values \cite{SVS04}. Another scenario 
of an intervention is a company (e.g., a bank or an airline \cite{For18})
announcing its refusal to do business with the gun advocacy group NRA.
Such
advertising strategies can have a double effect of convincing potential
customers who share relevant values and antagonizing those who do not.

Our main results show that 
such interventions, even if 
intending mainly to increase sales and without direct intention to polarize,
can have a side
effect of increasing the extent of polarization in the society.
For example, it might be that, in a population with initial opinions distributed
uniformly, a number of interventions introduces some weak correlations.
In our model, these correlations can be profitably exploited by
advertisers in subsequent interventions. As a side effect, the interventions
strengthen the correlations and increase polarization.

For example, suppose that after various advertising campaigns,
we observe that people who tend to like item A (say, electric cars)
tend to be liberal, and people who like a seemingly unrelated item B (say, firearms) tend to be conservative.
This may result from the advertisers exploiting some obvious connections, e.g., between 
electric cars and responding to climate change, and between firearms and respect for the military.
Subsequently, future advertising efforts for electric cars may feature other values associated with liberals in America to appeal to potential consumers: an advertisement might show a gay couple driving to their wedding in an electric car. Similarly, future advertisements for firearms may appeal to conservative values for similar reasons. 
The end result can be
that the whole society becomes more polarized 
by the incorporation of political topics into advertisements.

Throughout the paper, we analyze properties of our model in a couple of scenarios. With respect to the interventions, we consider two scenarios:
either there is one entity (an \emph{influencer}) trying to persuade agents to
adopt their opinion or there are
two competing influencers pushing different agendas.
With respect to the time scale of intervations, we also consider two cases: 
the influencer(s) can apply arbitrarily many interventions, i.e., the \emph{asymptotic} setting, or they need to maximize influence with a limited number of interventions, i.e., the \emph{short-term} setting. The questions asked are: (i) What
sequence of interventions should be applied to achieve  the influencer's objective?
(ii) What are the computational resources needed to compute this optimal sequence?
(iii) What are the effects of applying the interventions on the population's opinion structure? 
We give partial answers to those questions.
The gist of them is that in most cases, applying
desired interventions increases the polarization of agents.

\subsection{Model definition}
\label{sec:model}

The formal definition of our model is simple. 
We consider a group of $n$ \emph{agents}, whose opinions are represented by $d$-dimensional unit vectors, where each coordinate corresponds to a topic. 
We will look into how those opinions change after receiving
a sequence of \emph{interventions}. Each intervention is also a unit
vector in $\R^d$, representing the opinion contained in a message that the
influencer (e.g., an advertiser) broadcast to the agents.
Our model features one parameter: $\eta > 0$, signifying how strongly
an intervention influences the opinions.

The interventions $v^{(1)},\ldots,v^{(t)},\ldots$
divide the process into discrete time steps. Initially,
the agents start with opinions $u_1^{(1)},\ldots,u_n^{(1)}$.
Subsequently, applying intervention $v^{(t)}$ updates
the opinion of agent $i$ from $u_i^{(t)}$ to $u_i^{(t+1)}$.

After each intervention, the agents update their
opinions by moving towards or away from the intervention vector,
depending on whether or not they agree with it
(which is determined by the
inner product between the intervention vector $v^{(t)}$ and the opinion vector),
and normalizing suitably.
The update rule is given by
\begin{align}\label{eq:main}
    u_i^{(t+1)}=\frac{w_i^{(t+1)}}{\left\|w_i^{(t+1)}\right\|}\;,
    \qquad\qquad\text{where}\qquad\qquad
    w_i^{(t+1)}=u_i^{(t)}+\eta\langle u_i^{(t)},v^{(t)}\rangle\cdot v^{(t)}\;.
\end{align}
We note that, by expanding out the definition of $w_i^{(t+1)}$,
\begin{align}
\label{eq:main3}
    \|w_i^{(t+1)}\|^2 = \langle w_i^{(t+1)}, w_i^{(t+1)} \rangle = 1 + (2\eta + \eta^2) \langle u_i^{(t)}, v^{(t)} \rangle^2
\end{align}
In particular, this implies
that $\|w_i^{(t+1)}\|\ge 1$, and consequently that $u_i^{(t+1)}$ is well-defined.
The norm in~\eqref{eq:main} and everywhere else throughout
is the standard Euclidean norm.
Note that applying $v^{(t)}$ or $-v^{(t)}$ to an opinion $u_i^{(t)}$
results in the same updated opinion $u_i^{(t+1)}$.

\subsection{Example}
\label{sec:examples}

To illustrate our model, let us consider an empirical example with $\eta=1$.
Suppose an advertiser is marketing a new product. The opinion of the population has four dimensions. The population consists of $500$ agents, each with initial opinions $u_i^{(1)}=(u_{i,1},u_{i,2},u_{i,3},0)\in\mathbb{R}^4$
subject to $u_{i,1}^2+u_{i,2}^2+u_{i,3}^2=1$. 
The opinion on the new product is represented by the fourth
coordinate, which is initially set to zero for all agents.
These starting opinions are sampled independently at random from the uniform distribution on
the sphere. A typical arrangement of initial opinions is shown under
$t=1$ in Figure~\ref{tab:example-one}.

Suppose the advertiser chooses to repeatedly apply an intervention that couples
the product with the preexisting opinion on the first coordinate.
More concretely, let the intervention vector be
\begin{align*}
v=\left(\beta,0,0,\alpha\right)\;,
\qquad\qquad\text{where}\qquad\qquad
\alpha=\frac{3}{4}\,,\beta=\sqrt{1-\alpha^2}\;.
\end{align*}
In that case, an application of the intervention $v$ to an opinion $u_i^{(1)}=(u_{i,1},u_{i,2},u_{i,3},0)$
results in $\langle u_i^{(1)},v\rangle=\beta u_{i,1}$ and
\begin{align*}
    u_i^{(2)}=\frac{w_i^{(2)}}{\|w_i^{(2)}\|}\;,\qquad\qquad
    w_i^{(2)} = \big((1+\beta^2)u_{i,1}, u_{i,2}, u_{i,3}, \beta\alpha u_{i,1}\big)\;,
    \qquad\qquad
    \|w_i^{(2)}\|^2=1+3\beta^2 u_{i,1}^2\;.
\end{align*}
Note that after applying the intervention the first and last coordinates have the same sign.
In subsequent time step, the intervention $v$ is applied again to the updated opinions $u_i^{(2)}$
and so on.

The evolution of opinions over five consecutive 
applications of $v$ in this process is illustrated in Figure~\ref{tab:example-one}.
The interventions
increase the affinity for the product for some agents while antagonizing
others. Furthermore, they have a side effect of polarizing the agents' opinions 
also on the first three coordinates. A similar example is included
in Appendix~\ref{sec:example-two}.

\begin{figure}[!htp]
\begin{tabular}{cc}
   \multicolumn{1}{l}{$t=$ \num} & 
   \multicolumn{1}{l}{$t=$ \num}\\
  {\includegraphics[width=0.45\textwidth]{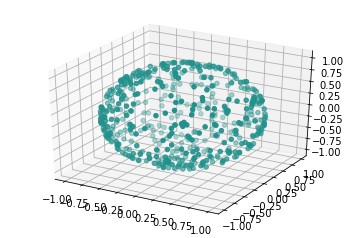}}&
    {\includegraphics[width=0.45\textwidth]{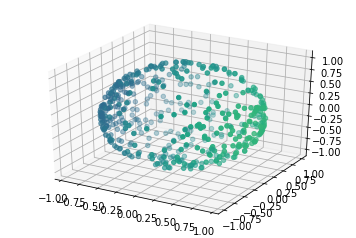}} \\
  \multicolumn{1}{l}{$t=$ \num} & 
  \multicolumn{1}{l}{$t=$ \num}\\
  {\includegraphics[width=0.45\textwidth]{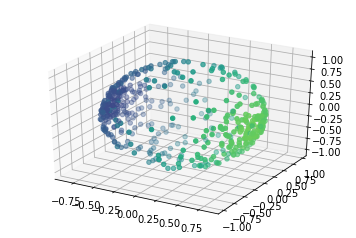}}
  &
    {\includegraphics[width=0.45\textwidth]{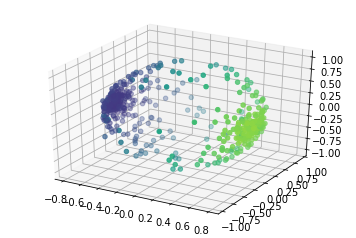}} \\
   \multicolumn{1}{l}{$t=$ \num} & 
   \multicolumn{1}{l}{$t=$ \num}\\
  {\includegraphics[width=0.42\textwidth]{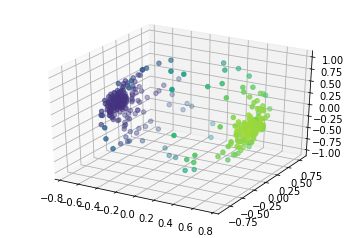}}
  &
    {\includegraphics[width=0.42\textwidth]{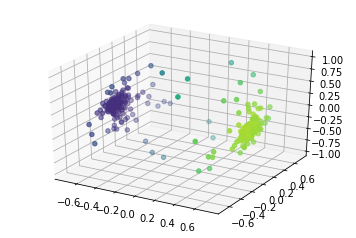}} \\
  \multicolumn{2}{c}{
  {\includegraphics[width=0.9\textwidth]{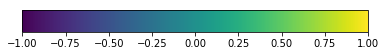}}}
\end{tabular}
\caption{Graphical illustration of the example discussed in Section~\ref{sec:examples}.
Since we are working in $d=4$, we illustrate the first three dimensions as spatial positions
and the fourth dimension with a color scale. Initially the opinions are uniformly distributed
on the sphere, with the fourth dimension equal to 0 (no opinion) everywhere.
Consecutive applications of the intervention $v=(\sqrt{7}/4,0,0,3/4)$ in $\R^4$ result in polarization
both in spatial dimensions and in the color scale.}
\label{tab:example-one}
\end{figure}

\subsection{Outline of our results}

We analyze the strategy of influencers in several settings. 

In an \textbf{``asymptotic scenario''},
the influencer wants to apply an infinite sequence of interventions
$v^{(1)},v^{(2)},\ldots,$ that maximizes
how many out of the $n$ agent opinions converge to the target vector
$v$. As is standard, we say that a sequence of vectors
$u^{(1)}, \cdots, u^{(t)},\ldots$ converges to a vector $v$ if
$\lim_{t\to\infty} ||u^{(t)} - v|| = 0$. One way to interpret this scenario
is that a campaigner wants
to establish a solid base of support for their party
platform.

In a \textbf{``multiple-influencer scenario"}, two influencers
(such as two companies or two parties) who have different objectives apply their
two respective interventions on the population in a certain order. We ask how the opinions change under such competing influences.
This scenario can be interpreted as two parties campaigning their agendas to
the population.

In a \textbf{``short-term scenario''},
the influencer is advancing a product/subject
which is expressed in the last coordinate
of opinion vectors $u_{i,d}$.
The influencer assumes some fixed threshold $0 < T < 1$ 
and an upper bound $K$
on the number of interventions, and asks, given $n$
opinions $u_1, \ldots, u_n$, how to
choose $v^{(1)},\cdots,v^{(K)}$ in order to maximize the number
of time-$(K+1)$ opinions $u_{1}^{(K+1)}, \ldots, u_n^{(K+1)}$ with $u_{i,d}^{(K+1)} > T$.
One interpretation is that advertisers only have a limited number
of opportunities to publicize their products to consumers,
and consumers with $u_i^{(K+1)} > T$ will decide to buy
the product after the interventions $v^{(1)},\cdots,v^{(K)}$ are applied.

\medskip

We briefly summarize our results for these scenarios.
In Section~\ref{sec:random-polarization} we start by showing that random interventions
lead to a strong form of polarization.
More precisely, assuming uniformly distributed initial opinions,
we prove that applying an independent uniformly
random
intervention at each time step
leads the opinions to form two equally-sized clusters converging to a pair of 
(moving) antipodal points. 

In Section~\ref{sec:densest-hemisphere-intro} we consider the asymptotic scenario,
where there is one influencer with a desired campaign agenda $v$ and unlimited numbers of interventions at its disposal.
We ask which sequence of interventions maximizes the number of opinions that
converge to the agenda $v$. Somewhat surprisingly,
we show that such optimal strategy does not necessarily promote the 
campaign agenda directly at every step. 
Instead, it finds a hemisphere containing the largest number of initial opinions, 
concentrates the opinions in this hemisphere
around an arbitrary point, and only in the last stage nudges them gradually towards the target agenda. 
We then show that it is computationally hard to approximate this
densest hemisphere (and therefore the optimal strategy) to any constant factor.
Again, strong polarization emerges from our dynamic: there exists a pair of antipodal points
such that all opinions converge to one of them. 

In Section~\ref{sec:short-term intro} we study
the short-term scenario where one influencer is allowed only one intervention.
In Section~\ref{subsec:short-term-1}, 
we describe a case study with one influencer and two agents in the population. We assume that the influencer wants
to increase the correlations of agent opinions with the target opinion $v$ above a given
threshold $T>0$. We show consequences of optimal interventions depending on if the influencer
can achieve this objective for one or both agents.
In Section~\ref{subsec:short-term-2}, we consider a similar scenario, but with a large number
of agents. In that case, it surprisingly turns out that the problem of finding optimal intervention
in this short-term setting is related to the problem analyzed in the asymptotic setting.
Finding the optimal intervention is equivalent to finding a spherical cap containing the largest number of initial opinions.

In Section~\ref{sec:two-advertisers-intro}, we study
two competing influencers.
At each time step, one of the influencers is selected at
random to apply its intervention.
One might hope that having multiple advertisers can make the resulting opinions more spread-out, but we prove that this not the case.
We show that, as time goes to infinity, 
all opinions converge to the convex cone 
between the two intervention vectors.
Furthermore, we show that the 
if the correlation between the interventions is high enough,
the strong form of polarization emerges: the opinions of the population concentrate around two antipodes moving around in the convex cones of the two interventions. 

\subsection{Design choices}
\label{sec:design-choices}
Our goal in this work is to provide a simple, elegant and analyzable model demonstrating how correlations between different topics and natural interventions lead to polarization. That being the case, there are many societal mechanisms related to polarization that we do not discuss here.

First, in contrast to majority of existing
literature, we present a mechanism independent from opinion changes induced by interactions between individuals.
Second, we do not address aspects such as replacement
of the population or unequal exposure and effects of the interventions.
We do not consider any external influences on the population in addition to
the interventions.
Our model does not align with (limited) theoretical and empirical
research suggesting that in certain settings
exposure to conflicting views can decrease polarization
\cite{PT06, MS10,GDFM17,GGPT17}
or works that question the overall extent of polarization in the society
\cite{FAP05,BG08}.

In general we assume that the influencers have full knowledge of the agent opinions. This is not a realistic
assumption and in fact our results in
Section~\ref{sec:densest-hemisphere-intro} show that in some
settings the optimal influencer strategy is infeasible to compute
even with the full knowledge of opinions. On the other hand,
we observe polarization also in settings where the influencers
apply interventions that are agnostic to the opinions,
for example with purely random interventions in Section~\ref{sec:random-polarization} or competing
influencers in Section~\ref{sec:two-advertisers-intro}.

We sometimes discuss the uniform distribution of initial opinions on $\R^d$.
We do this as the uniform distribution may be viewed as the most diverse and establishing polarization starting from the uniform distribution hints that we are modelic a generic phenomenon. 
Most of our results do not make assumptions about the initial distribution.

We assume that any group of topics can be combined into an intervention
with the effect given by~\eqref{eq:07}. A more plausible model
might feature some ``internal'' (content) correlations between topics
in addition to ``external'' (social) correlations arising out of the agents' opinion
structure.
For example, topics may have innate connections, causing
inherent correlations between corresponding opinions
(e.g., being positive
on renewable energy and recycling).
Furthermore, there are certain topics (e.g., undesirability of murder)
on which nearly all members of the population share the same inclination.
As a matter of fact, it is common for marketing strategies to exploit
unobjectionable social values (see, e.g.,~\cite{VSL77}).
However, we presume that under suitable circumstances (e.g., due to
inherent correlations we just mentioned) the ``polarizing''
topics might present a more appealing alternative for a campaign.
Our model concerns such a case, where the
``unifying'' topics might be excluded from the analysis. We note that other works have also suggested that focusing on polarizing topics may be appealing for campaigns~\cite{PY18}.

Below we discuss a couple of specific design choices in more detail:

\paragraph{Euclidean unit ball}
We make an assumption that all opinions and interventions lie on the
Euclidean unit ball. 
Note that the interpretation of this representation is somewhat ambiguous.
The magnitude of an opinion on a given subject $u_{i,k}$ might signify
the strength of the opinion, the confidence of the agent or relative importance
of the subject to the agent. 
While these are different measures,
there are psychological reasons to expect
that, e.g., ``issue interest'' and ``extremity of opinion''
are correlated \cite{LBS00,Bal07,BB07}.
Especially taking the magnitudes as signifying the relative importance, we believe that
the assumption that this ``budget of importance'' for any given agent 
is fixed is quite natural. That being said, we are also motivated
by simplicity and tractability.

Multiple ways of relaxing or modifying this assumption
are possible. While we do not study these variants in
this paper, we now discuss them very briefly. At least empirically, our basic findings about ubiquity of polarization seem to
remain valid for those modified models.

Perhaps the simplest modification is to use the same update rule
as in~\eqref{eq:main} with a different norm (eg., $\ell_1$ or
$\ell_\infty$ norm). Such variant would also assume that opinions
and interventions lie on the unit sphere of the respective norm.
Our experiments suggest that, qualitatively, 
both $\ell_1$ and $\ell_\infty$ variants behave similarly
to the Euclidean norm. 

In another direction,
rather than having all opinions on the unit sphere,  
fixed, but different norms $z_i$ 
can be specified for different agents.
Then, the update rule~\eqref{eq:main} could be modified as
\[
    w_i^{(t+1)} = u_i^{(t)} + \eta\cdot\left\langle\frac{u_i^{(t)}}{z_i},
    v^{(t)}
    \right\rangle\cdot v^{(t)}\;,
\]
with normalization preserving $\|u_i^{(t+1)}\|=z_i$.
As long as the values of $z_i$ are bounded from below and above, 
the resulting dynamic is essentially identical and our results carry over to this more general setup. 

Yet another possibility is to consider opinion unit vectors $u\in\mathbb{R}^{d+1}$ with $u_{d+1}\ge 0$ and interpret
the first $d$ coordinates as opinions and the last coordinate
as ``unused budget''. Therefore, large values of $u_{d+1}$
signify generally uncertain opinions and small values
of $u_{d+1}$ correspond to strong opinions.
There are multiple possible rules for interventions,
where an intervention can have $d$ or $d+1$ coordinates, and
with different treatments of the last coordinate. We leave
the details for another time.



\paragraph{Effects of applying $v$ and $-v$} 
In our model, an effect of an intervention $v$
is exactly the same as for the opposite intervention $-v$. This might look like
a cynical assumption about human nature, but arguably it is not
entirely inaccurate. For example,
experiments on social media show that not only exposure to similar ideas (the ``echo chamber'' effect), but also exposure to opposing opinions
causes beliefs to become more polarized \cite{BAB18}.
This is even more apparent if a broader notion of an 
intervention is considered. Using a recent example, social media platforms banning
or disassociating from certain statements can have a polarizing 
effect~\cite{bbc20}.
Furthermore, in our model this effect occurs only if all the
components of an opinion are negated.

A related, more general objection is that direct persuasion is not
possible in our model. If an agent has an opinion $u$ with
$\langle u,v\rangle<0$, directly applying $v$ only makes the situation
worse. Instead, an effective influencer needs to apply interventions utilizing
different subjects to gradually move $u$ through a sequence of intermediate
positions towards $v$. Our answer is that we posit that a lot of, if not all, persuasion actually works that way: to convince that ``$x$ is good'',
one argues that ``$x$ is good, since it is quite like $y$, which we 
both already agree is good''.

\paragraph{Notions of polarization} 
While the notion of polarization is clear when discussing one topic, it is not straightforward to interpret in higher dimensions. Let $S\subseteq\mathbb{R}^d$ be
a set of $n$ opinions. Writing
$u=(u_1,\ldots,u_d)$ for $u\in S$,
a natural measure of polarization of $S$ 
on a single topic $i$ is 
\[
\rho_i(S) = \frac{1}{|S|^2} \max_{T \subset S} \sum_{u \in T, u' \in S \setminus T} (u_i - u'_i)^2, 
\]
and we may generalize it to higher dimensions by measuring the polarization as:
\[
\rho(S) = \frac{1}{|S|^2} \max_{T \subset S} \sum_{u \in T, u' \in S \setminus T} \| u - u' \|^2. 
\]
It is clear from the definition that 
\[
\max_i \rho_i(S) \leq \rho(S) \leq \sum_i \rho_i(S). 
\]
If we consider an example set $S_1$ with $n/2$ opinions at $u$ and $n/2$ opinions at $-u$, then clearly 
$\rho(S_1) = \sum_i \rho_i(S_1)$, but in any other example, the upper bound will not be tight. 
For example, if $S_2$ is the set of the $2^d$ vertices of a hypercube, i.e., $S_2=1/\sqrt{d}\cdot\{-1,1\}^d$, then $\rho_i(S_2) = 1/d$ for all $i$, but $\rho(S_2)$ converges to $1/2$ as $n\to\infty$. 
This corresponds to the fact that while the society is completely polarized on each topic, 
two random individuals will agree on about half of the topics. In Section~\ref{sec:related} we refer to such a situation
as exhibiting \emph{issue radicalization}, but no
\emph{issue alignment}.

Ultimately, in many of our results we do not worry about these issues,
since we observe
a strong form of polarization, where all opinions
converge to two antipodal points.

\subsection{Other variants}

Other than discussed above, there are many possible variants that can lead to interesting future work. These include: 
\begin{itemize}
    \item ``Targeting'', where the influencer can select subgroups of the population and apply interventions groupwise.
    \item Models where the strength of an intervention
    $\eta$ varies across agents and/or time steps.
    \item Perturbing preferences with noise after each step.
    \item Replacement of the population, e.g., introducing new agents
    with ``fresh'' opinions or removing agents that stayed in the
      population for a long time or who already ``bought'' the product, i.e.,
      exceeded the threshold $u_{i,d}>T$. For example, this could correspond to "one-time" purchase product like a house or a fridge, or situations where the customer's opinion is more difficult to change
      as time passes.
      \item Models where the initial opinions are not observable or partially observable.
      \item Expanding the model by adding peer effects and social network structure and exploring the resulting dynamics of polarization and opinion formation. This can be done in
      different ways and we expect that polarization will
      feature in many of them. For example,
      \cite{GKT21} show polarization for random interventions
      in what they term the ``party model''.

    \item Strategic competing influencers: in the studied scenarios with competing influencers, we assume that they apply fixed interventions. One can ask: suppose the influencers have their own target opinions, what is each campaigner's optimal sequence of messages in face of the other campaigner? Then, resulting equilibrium of opinion formation could be analyzed. This can be modeled as a dynamic game where the game state is the opinion configuration and optimal strategies may be derived using sequential planning and control.
\end{itemize}

\section{Related works}
\label{sec:related}
As mentioned, there is a multitude of modeling and empirical works studying opinion polarization in different contexts
\cite{NSL90, Axe97, BG98, Noa98, HK02, MKFB03, Mul05, BB07, DGL13, DVSC17, KP18, SCP19,  PY18, BAB18}. Broadly speaking, previous works have proposed various possible sources for polarization, including peer interactions, bias in individuals' perceptions, and global information outlets. 

There is an extensive line of models of opinion exchange on networks with peer
interactions, where individuals encounter  neighboring individuals' opinions and
update their own opinions based on, e.g.,
pre-defined friend/hostile relations \cite{SPJBJ16},
or the similarity and relative strength of opinions \cite{MS10}, etc.
This branch of work often attributes polarization to homophily of one's
social network \cite{DGL13} that is induced by the self-selective nature of social relations and segregation
of like-minded people \cite{WMKL15} and
exacerbated by the echo chamber effect of social media \cite{Pariser}. 

A parallel proposed mechanism points to psychological biases in individuals' opinion formation processes. One example is biased assimilation \cite{LRL79,DGL13,BB07, BAB18}:
the tendency to reinforce one's original opinions regardless
if other encountered
opinions align with them or not. For example, \cite{BAB18} observed that even when social media users are assigned to follow accounts that share opposing opinions, they still tend to hold their old political opinions and often to a more extreme degree. On the modeling side, \cite{DGL13} showed that DeGroot opinion dynamics with the biased assimilation property on a homophilous network may lead to polarization.

Existing works have also proposed models where polarization occurs even when information is shared globally \cite{Zaller, Mul05}.
For example, \cite{Mul05} propose a model where competition for readership between global information outlets causes news to become polarized in a single-dimensional setting. 
Another example is \cite{Zaller}, a classical work on the formation of mass opinion. It theorizes that each individual has political dispositions formed in their own life experience, education and previous encounters that intermediate between the message they encounter and the political statement they make. Therefore, hearing the same political message can cause different thinking processes and changes in political preferences in different individuals. 

It is noteworthy that the majority of previous work focuses on polarization on a single topic dimension. Two exceptions are \cite{BB07}, which studies biased assimilation with opinions on multiple topics and \cite{BG08} that observed non-trivial correlations between people's attitudes on different issues. 
We note that \cite{BB07} uses a different updating rule to observe
dynamics that differ from our work: in their simulations, polarization on one issue
typically does not result in polarization on others. There is also a class of models \cite{Axe97, Noa98, MKFB03}
that concern multi-dimensional opinions where an opinion on a given topic
takes one of finitely many values (e.g., + or \textminus).
These models do not seem to have a geometric structure of opinion space similar to ours
and usually focus on formation of discrete groups
in the society rather than total polarization.
Another model~\cite{PPTF17} uses a geometric (affine) rule of
updating multi-dimensional opinions. Unlike us, they seem to be modeling
pre-existing, ``intrinsic'' correlations between topics rather than the emergence
of new ones and they are concerned mostly with convergence and
stability of their dynamics.
  
A related paper \cite{PY18} contains a geometric model of opinion (preference) structures.
Both this and our model propose mechanisms through which information outlets acting for their own benefit can lead to increased disagreement in the society. 
The key difference to our model is that their population's preferences are static and do not update, but the 
outlets are free to choose what information to
offer to their customers. 
By contrast, in our model, the influencers have pre-determined ideologies and compete to align agents' opinions with their own. In other words, \cite{PY18} focuses on modeling of competitive information acquisition, and our paper on modeling the influence of marketing on the public opinion.

Our model suggests that
under the
conditions
of biased assimilation, opinion manipulation by one or several
global information outlets can unintentionally lead to
a strong form of polarization in multi-dimensional opinion space. 
Not only do people polarize on individual issues, but also their opinions on previously unrelated issues become correlated. 
This form of polarization is known as \emph{issue alignment} \cite{BG08} in political science and sociology literature. Issue alignment refers to an opinion structure where the population's  opinions on multiple (relatively independent) issues correlate.
It is related to \emph{issue radicalization}, where the opinions polarize for each issue separately.
Compared to issue radicalization, issue alignment is theorized to pose more constraints on the opinions an individual can take, resulting in polarized and clustered mass opinions even when the public opinions are not extreme in any single topic, and presenting
more obstacles for social integration and political stability \cite{BG08}. 
In light of this, one way to view our model is as a
mathematical mechanism by which this strong form of polarization can arise and
worsen due to companies', politicians', and the media's natural attempts to gain
support from the public.

On the more technical side, we note that our update equation bears similarity to Kuramoto model \cite{JMB04} for synchronization of  oscillators on a network in the control literature.
In this model, each oscillator $i$ is associated with the point $\theta_i$ on the two-dimensional sphere, and $i$ updates its point continuously as a function of its neighbors' points $\theta_j$:

$$\dot{\theta_i} = \omega_i +\frac{K}{N}\sin(\theta_j-\theta_i),$$
where $K$ is the \emph{coupling strength} and $N$ is the number of nodes in the network. In two dimensions, our model can be compared to Kuramoto model with $\omega_i=0$ on a star graph, with the influencers at the center of the star connected to the entire population, where the influencers' opinions do not change and the update strength is qualitatively similar to $\sin((\theta_v-\theta_u)/2)$ (see ~\eqref{eq:pull_function}).
However, we note a crucial difference: in the Kuramoto dynamic, $\theta_i$ always moves towards $\theta_j$, i.e. nodes always move towards synchronization, but in our dynamic, opinions $\theta_i$ are allowed to move further away from $\theta_j$ when the angle between their opinions are obtuse. In addition, the central node in our model can be strategic in choosing its positions, while the central node in Kuramoto model follows the synchronization dynamics of the system. We think this property provides a better model for opinion interactions.

\paragraph{Subsequent work} A work 
by Gaitonde, Kleinberg and Tardos~\cite{GKT21},
announced after we posted the preprint of
this paper, proposes a framework that generalizes our random
interventions scenario from Section~\ref{sec:random-polarization}.
They prove several interesting results, including a strong form of polarization under random interventions in some related models.
They also shed more light on the scenario of dueling
influencers from Section~\ref{sec:two-advertisers-intro},
showing that in case the dueling interventions are orthogonal,
the resulting dynamics exhibits a weaker kind of polarization.

\section{Asymptotic scenario: random interventions polarize opinions}
\label{sec:random-polarization}

In this section, we analyze the long-term behavior of our model in a simple random setting.
We assume that, for given dimension $d$
and parameter $\eta$, at the initial time $t=1$ we are given $n$ opinion vectors
$u^{(1)}_1,\ldots,u^{(1)}_n$. Subsequently, we
sample a sequence of 
interventions $v^{(1)},v^{(2)},\ldots,$ each
$v^{(t)}$ iid from the uniform distribution on
the unit sphere $S^{d-1}$. At time $t$ we apply the random intervention $v^{(t)}$ to every opinion vector $u_i^{(t)}$, obtaining a new opinion $u_i^{(t+1)}$.

We want to show that the opinions $\{u_i^{(t)}\}$
almost surely polarize as time $t$ goes to infinity.
We need to be careful about defining the notion of polarization:
since the interventions change at every time step, 
the opinions cannot converge to a fixed vector. Instead,
we show that for every pair of opinions the angle between them
converges either to $0$ or to $\pi$. More formally:

\begin{theorem}\label{thm:random-polarization}
  Consider the model of iid interventions described above for some $d\ge 2$, $\eta>0$ and initial opinions $u_1^{(1)},\ldots,u_n^{(1)}$.
  For any $1\le i<j\le n$ and
  $t\to\infty$,
\begin{align*}
      \Pr\left[ 
      \|u_i^{(t)}-u^{(t)}_j\|\to 0\lor\|u_i^{(t)}+u_j^{(t)}\|\to 0  
      \right] = 1 \;.
  \end{align*}
\end{theorem}

This leads to a corollary
which follows
by applying the union bound (with probability 0
in each term) for each pair of opinions
$u_i^{(t)},u_j^{(t)}$:
\begin{corollary}\label{cor:random-polarization}
    For any $d\ge 2$, $\eta>0$, initial opinions $u_1^{(1)},\ldots,u_n^{(1)}$
    and a sequence of uniform iid interventions,
    almost surely, there exists $J\subseteq\{1,\ldots,n\}$
    such that the diameter of the set
\begin{align*}
    \left\{
    (-1)^{\mathbbm{1}[i\in J]} \cdot u^{(t)}_i:
    i\in\{1,\ldots,n\}
    \right\}
\end{align*}
converges to zero.
\end{corollary}

\begin{remark}
  Consider initial opinions of $n$ agents that are independently 
  sampled from a distribution $\Gamma$ that is symmetric around the origin, in the sense that 
  $\Gamma(-A) = \Gamma(A)$ for every set $A\subseteq S^{d-1}$. 
  Then, with high probability, 
  the opinions converge to two polarized clusters of size
  roughly $n/2$. Indeed, consider sampling $n$ independent vectors $u_1,\ldots,u_n$ from $\Gamma$ and $n$ independent 
  signs $\sigma_1,\ldots,\sigma_n \in \{ \pm 1\}$. Then $\sigma_1 u_1, \ldots, \sigma_n u_n$ are independent samples from $\Gamma$. Moreover, if the sizes of the two clusters for $u_1,\ldots,u_n$ are $r$ and $n-r$ then the size of each cluster for $\sigma_1 u_1, \ldots, \sigma_n u_n$ is distributed
  according to $\mathrm{Bin}(r,1/2) + \mathrm{Bin}(n-r,1/2) = \mathrm{Bin}(n,1/2)$ (this is due to the observation
  that $u_i$ and $-u_i$ converge to the opposite clusters).
\end{remark}

\begin{remark}
For simplicity we do not elaborate on this later,
but we note that, both empirically and theoretically, the convergence in our results is
quite fast. This concerns Theorem~\ref{thm:random-polarization}, as well as 
the results presented in the subsequent sections.
\end{remark}

\smallskip

We now proceed to the proof of Theorem~\ref{thm:random-polarization}:
\subsection{Notation and main ingredients}

Before we proceed with explaining the proof, let us make a general observation
that we will use frequently. Let $d\ge 2$ and $\eta>0$ and let
$f:S^{d-1}\times S^{d-1}\to S^{d-1}$ be the function
mapping an opinion $u$ and an intervention $v$
to an updated opinion $f(u, v)$, according
to~\eqref{eq:main} and~\eqref{eq:main3}. It should be clear that
this function is invariant under isometries: namely, for any 
real unitary
transformation $A:S^{d-1}\to S^{d-1}$ we have
\begin{align}\label{eq:09}
    f(Au, Av) = Af(u, v)\;.
\end{align}
In our proofs we will be often using~\eqref{eq:09} to choose a convenient 
coordinate system. 

Let us turn to Theorem~\ref{thm:random-polarization}.
Again, let $d\ge 2$ and $\eta>0$. Without loss of generality we will consider 
only two starting opinions called $u_1^{(1)}$
and $u_2^{(1)}$.
To prove Theorem~\ref{thm:random-polarization}, we need to show that
almost surely one of the vectors
$u_1^{(t)}-u_2^{(t)}$ and $u_1^{(t)}+u_2^{(t)}$
vanishes.

We proceed by using martingale convergence. Specifically,
let
\begin{align*}
    \alpha_t := \arccos\langle u_1^{(t)},u_2^{(t)}\rangle\;.
\end{align*}
That is, $0\le\alpha_t\le\pi$ is the primary angle between $u_1^{(t)}$ and $u_2^{(t)}$. 

The proof rests on
two claims. First, $\alpha_t$ is a 
martingale:
\begin{claim}\label{cl:martingale}
  $\EE[\alpha_{t+1}\mid\alpha_t] = \alpha_t$.
\end{claim}
Second, we show a property
which has been called 
``variance in the middle''~\cite{BGN18}:
\begin{claim}\label{cl:variance}
  For every $\eps>0$, there exists $\delta>0$ such that,
  \begin{align}\label{eq:12}
  \eps\le\alpha_t\le\pi/2\implies
  \Pr\big[\alpha_{t+1}<\alpha_t-\delta\mid \alpha_t\big]>\delta\;,
  \end{align}
  and, symmetrically,
  \begin{align}\label{eq:13}
  \pi/2\le\alpha_t\le\pi-\eps\implies
  \Pr\big[\alpha_{t+1}>\alpha_t+\delta\mid\alpha_t\big]>\delta\;.
  \end{align}
\end{claim}

These two claims imply Theorem~\ref{thm:random-polarization}
by standard tools from the theory of martingales
(eg.,~\cite{Wil91}):
\begin{proof}[Claims~\ref{cl:martingale}
and~\ref{cl:variance} imply Theorem~\ref{thm:random-polarization}]
As a consequence of applying Claim~\ref{cl:variance} $\lceil\pi/\delta\rceil$ times, we obtain that
for every $\eps>0$ there exist $k_0\in\mathbb{N}$ and $\eta<1$
such that
\begin{align*}
    \eps\le\alpha_t\le\pi-\eps\implies
    \Pr\left[\forall 1\le k\le k_0: \eps\le\alpha_{t+k}\le\pi-\eps\mid
    \alpha_t\right]\le\eta\;.
\end{align*}
Subsequently, it follows that for any fixed $\eps>0$ and
$T\in\mathbb{N}$,
\begin{align}\label{eq:14}
    \Pr\left[\forall t\ge T:\eps\le\alpha_{  t}\le\pi-\eps\right]
    =0\;.
\end{align}

By Claim~\ref{cl:martingale}, the sequence of
random variables $\alpha_t$ is a bounded martingale and therefore almost
surely converges. Accordingly, let 
$\alpha^*:=\lim_{t\to\infty}\alpha_t$. We now argue that $\Pr[0<\alpha^*<\pi]=0$. To that end,
\begin{align*}
    \Pr[0<\alpha^*<\pi]&\le
    \sum_{s=1}^{\infty}
    \Pr\left[\frac{1}{s}<\alpha^*<\pi-\frac{1}{s}\right]
    \le\sum_{s=1}^{\infty}
    \Pr\left[\exists T:\forall t\ge T:
    \frac{1}{2s}<\alpha_t<\pi-\frac{1}{2s}\right]\\
    &\le\sum_{s=1}^{\infty}\sum_{T=1}^{\infty}
    \Pr\left[\forall t\ge T:\frac{1}{2s}<\alpha_t<\pi-\frac{1}{2s}
    \right]=0\;,
\end{align*}
where we applied~\eqref{eq:14} in the last line. Hence,
almost surely, either $\alpha^*=0$, which is equivalent to
$\|u_1^{(t)}-u_2^{(t)}\|\to 0$ or $\alpha^*=\pi$, equivalent
to $\|u_1^{(t)}+u_2^{(t)}\|\to 0$.
\end{proof}

In the subsequent sections we proceed with proving Claims~\ref{cl:martingale}
and~\ref{cl:variance}. In Section~\ref{sec:cl-martingale-d-2}
we prove Claim~\ref{cl:martingale} for $d=2$. In 
Section~\ref{sec:cl-martingale-d-3} we show the same claim for $d\ge 3$
by a reduction to the case $d=2$. Finally, in Section~\ref{sec:cl-variance-proof}
we use a continuity argument to prove Claim~\ref{cl:variance}.

In the following proofs, we fix $d$, $\eta$, time $t$
and the opinions of two agents at that time. For simplicity,
we will denote the relevant vectors as $u:=u_1^{(t)}$, $u':=u_2^{(t)}$
and $v:=v^{(t)}$.

\subsection{Proof of Claim~\ref{cl:martingale} for 
\texorpdfstring{$d=2$}{d=2}
}
\label{sec:cl-martingale-d-2}

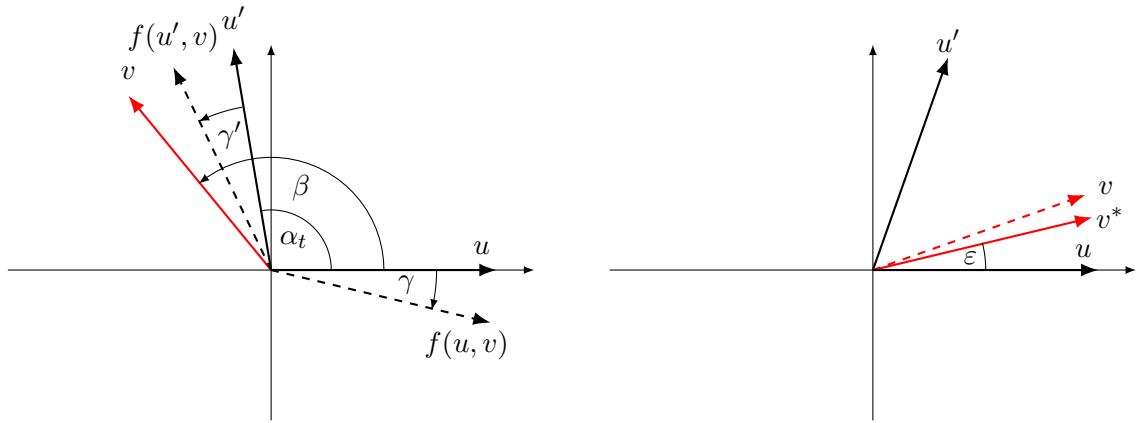
\begin{figure}[!ht]\centering\begin{tikzpicture}
    \draw [-latex] (0, -2) -- (0, 3);
    \draw [-latex] (-3.5, 0) -- (3.5, 0);
    \draw [-latex] (8, -2) -- (8, 3);
    \draw [-latex] (4.5, 0) -- (11.5, 0);

    \draw [-Latex, thick, color=red]
    (0, 0) -- (-1.9, 2.32);
    \draw [-Latex, thick, dashed]
    (0, 0) -- (-1.3, 2.70);
    \draw [-Latex, thick] (0, 0) -- (-0.5, 2.96);
    \draw [-Latex, thick] (0, 0) -- (3,0);
    \draw [-Latex, thick, dashed]
    (0, 0) -- (2.92, -0.7);
    
    \node at (-1.9, 2.62) {$v$};
    \node at (-1.3, 3.10) {$f(u',v)$};
    \node at (-0.5, 3.36) {$u'$};
    \node at (2.8, 0.3) {$u$};
    \node at (2.6, -1) {$f(u,v)$};
    
    \draw [domain=0:130, -latex] plot
    ({1.5*cos(\x)},{1.5*sin(\x)});
    \draw [domain=0:100] plot ({0.8*cos(\x)},{0.8*sin(\x)});
    \draw [domain=100:116, -latex] plot ({2.2*cos(\x)},{2.2*sin(\x)});
    \draw [domain=0:-14, -latex] plot ({2.2*cos(\x)},{2.2*sin(\x)});
    
    \node at (0.3, 0.4) {$\alpha_t$};
    \node at (0.4, 1.1) {$\beta$};
    \node at (-0.55, 1.8) {$\gamma'$};
    \node at (1.8, -0.2) {$\gamma$};
    
    \draw [-Latex, thick] (8, 0) -- (9, 2.83);
    \draw [-Latex, thick, dashed, color=red]
    (8, 0) -- (10.83, 1);
    \draw [-Latex, thick, color=red]
    (8, 0) -- (10.92, 0.7);
    \draw [-Latex, thick] (8, 0) -- (11,0);
    
    \node at (9, 3.05) {$u'$};
    \node at (11.1, 1.1) {$v$};
    \node at (11.15, 0.7) {$v^*$};
    \node at (10.8, 0.25) {$u$};
    
    \draw [domain=0:14] plot
    ({8.0+1.5*cos(\x)},{1.5*sin(\x)});
    
    \node at (9.3, 0.16) {$\eps$};
  \end{tikzpicture}
  \caption{On the left an illustration of the vectors
  and angles in the proof of Claim~\ref{cl:martingale}.
  On the right an illustration for the proof of
  Claim~\ref{cl:variance}.}
  \label{fig:random-intervention}
\end{figure}

    It follows from~\eqref{eq:09} that we can assume
    wlog that $u=(1,0)$ and $u'=(\cos\alpha_t,\sin\alpha_t)$ (recall
    that by definition $0 \leq \alpha_t\le\pi$ holds). 
    Let us write the random intervention vector as $v=(\cos\beta,\sin\beta)$, where    
    the distribution of $\beta$ is uniform
    in $[0,2\pi)$. We will also write
    (cf.~Figure~\ref{fig:random-intervention}
    for an overview)
    \begin{align*}
        f(u,v)=(\cos\gamma,\sin\gamma)\;,\qquad\qquad
        f(u',v)=(\cos(\alpha_t+\gamma'),
        \sin(\alpha_t+\gamma'))\;,\qquad\qquad
        \gamma,\gamma'\in[-\pi,\pi)\;.
    \end{align*}
    
    Note that $\gamma$ is a function of the intervention
    angle $\beta$, and it should be clear that
    $\gamma(\beta)=-\gamma(-\beta)$. Accordingly,
    the distribution of $\gamma$ is symmetric around zero
    and in particular $\EE\gamma=0$ (where the expectation
    is over $\beta$). Applying~\eqref{eq:09}, it also follows
    $\EE\gamma'=0$. 

    Let $\alphawhat:=\alpha_t+\gamma'-\gamma$.
    Since $\alphawhat$ is equal to 
    the directed angle from $f(u,v)$ to $f(u',v)$, 
    one might think that
    we have just established $\EE[\alpha_{t+1}\mid\alpha_t]=\alpha_t$.
    However, recall that we defined $\alpha_{t+1}$ 
    to be the value of the
    (primary) \emph{undirected} angle between $f(u,v)$ and $f(u',v)$.
    In particular, it holds that $0\le\alpha_{t+1}\le\pi$, but we cannot assume that about $\alphawhat$.
    On the other hand, it is clear that if 
    $0\le\alphawhat\le\pi$, then indeed 
    $\alpha_{t+1}=\alphawhat$.
    Therefore, in the following we will show
    that $0\le\alphawhat\le\pi$ always holds,
    which implies
    $\EE[\alpha_{t+1}\mid\alpha_t]=
    \EE[\alphawhat\mid\alpha_t]=\alpha_t+\EE\gamma'-\EE\gamma
    =\alpha_t$.
    
    To that end, we start with showing a weaker bound
    $-\pi<\alphawhat<2\pi$. To see this, we first establish that
    $-\pi/2<\gamma,\gamma'<\pi/2$. The argument for $\gamma$
    is as follows: if $\langle u,v\rangle\ge 0$, then
    $f(u,v)$ is a convex combination of $u$ and $v$. Therefore,
    an intervention cannot move $f(u,v)$ away from $u$ by
    an angle of more than $\pi/2$. If $v\ne u$, then
    also $f(u,v)\ne v$, so in fact the angle must be strictly
    less, that is $-\pi/2<\gamma<\pi/2$. If $\langle u,v\rangle<0$,
    then $-\pi/2<\gamma<\pi/2$ follows from the same argument
    applied to $-v$ (since the effect of both interventions
    is the same). Finally, $-\pi/2<\gamma'<\pi/2$ holds by~\eqref{eq:09} and the same proof.
    Since we know $0\le\alpha_t\le\pi$, we obtain
    $-\pi<\alphawhat<2\pi$.

    Since we know $-\pi<\alphawhat<2\pi$, the inequality
    $0\le\alphawhat\le\pi$ is equivalent to
    $\sin\alphawhat\ge 0$.
    Geometrically, this property means 
    that the ordered pair of vectors $(u,v)$
    has the same orientation as the pair $(f(u,v),f(u',v))$.
    To avoid case analysis, we prove this claim by a calculation:

\begin{claim}\label{cl:sin-calculation}
$\sin\alphawhat\ge 0$.
\end{claim}
\begin{proof}
We defer the proof to Appendix~\ref{app:sin-proof}.
\end{proof}

\subsection{Proof of Claim~\ref{cl:martingale}
for \texorpdfstring{$d\ge 3$}{d>=3}
}
\label{sec:cl-martingale-d-3}

In this case we will write the random intervention
vector as $v=v^{\parallel}+v^{\perp}$
where $v^{\parallel}$ is projection of $v$
onto the span of $u$ and $u'$. In particular,
$v^{\parallel}$ and $v^{\perp}$ are orthogonal.
We will now prove a stronger claim
$\EE\big[\alpha_{t+1}\mid\alpha_t,\|v^{\parallel}\|\big]=\alpha_t$.

Accordingly, condition on the value
$\|v^{\parallel}\|=R$.
Observe that, by symmetry, 
vector $v^{\parallel}$ is distributed uniformly in the
two-dimensional space $\Span\{u,u'\}$ among vectors of
norm $R$. In other words, we can write
$v^{\parallel} = RV$, where $V$ is a uniform
two-dimensional unit length vector.

Denote the non-normalized
vectors after intervention as
\begin{align*}
    \uhat:=u+\eta\langle u,v^{\parallel}\rangle(v^{\parallel}+
    v^{\perp})\;,\qquad\qquad
    \uhat':=u'+\eta\langle u',v^{\parallel}
    \rangle(v^{\parallel}+v^{\perp})
    \;.
\end{align*}

Let $c:=2\eta+\eta^2$. We proceed with calculations:
\begin{align*}
    \langle\uhat,\uhat'\rangle
    &=\langle u,u'\rangle+c\langle u,v^{\parallel}\rangle\langle u',v^{\parallel}\rangle
    =\langle u,u'\rangle+cR^2\langle u,V\rangle
    \langle u',V\rangle\;,\\
    \|\uhat\|^2
    &=1+c\langle u,v^{\parallel}\rangle^2=1+cR^2\langle u,V\rangle\;,\\
    \|\uhat'\|^2&=1+c\langle u',v^{\parallel}\rangle^2
    =1+cR^2\langle u',V\rangle\;.
\end{align*}
Note that all these formulas are valid also for
$d=2$, with the only difference that
$R=1$ holds deterministically in that case.

Since $c(\eta)=2\eta+\eta^2$ is a bijection on $\mathbb{R}_{>0}$,
there exists $\widehat{\eta}>0$ such that $cR^2=2\widehat{\eta}
+\widehat{\eta}^2$.
Accordingly, for any $d\ge 3$ and $\eta>0$,
the joint distribution of
$\langle\uhat,\uhat'\rangle$, $\|\uhat\|$ and $\|\uhat'\|$
conditioned on $\alpha_t=\arccos(\langle u,u'\rangle)$
and $\|v^{\parallel}\|=R$ is the same as their joint distribution
for $d=2$ and $\widehat{\eta}$, conditioned on the same
value of $\alpha_t$.

Since $\alpha_{t+1}=\arccos\left(
\frac{\langle\uhat,\uhat'\rangle}{\|\uhat\|\|\uhat'\|}\right)$,
the same correspondence holds for the distribution
of $\alpha_{t+1}$ conditioned on $\alpha_t$ and
$\|v^{\parallel}\|=R$. Therefore, 
$\EE\big[\alpha_{t+1}\mid \alpha_t,\|v^{\parallel}\|=R\big]=\alpha_t$
follows by Claim~\ref{cl:martingale} for $d=2$, which we
already proved.
\hfill\qedsymbol

\subsection{Proof of Claim~\ref{cl:variance}}
\label{sec:cl-variance-proof}
  Again we use~\eqref{eq:09} to choose a coordinate
  system and assume wlog that $u=(1,0,\ldots,0)$ and 
  $u'=(\cos\alpha_t,\sin\alpha_t,0,\ldots,0)$.
  Our objective is to show that, with probability at least
  $\delta$, we will have $\alpha_{t+1}-\alpha_t>\delta$
  (in case $\alpha_t\le\pi/2)$ or $\alpha_{t+1}-\alpha_t<-\delta$ 
  (in case $\alpha_t\ge\pi/2)$. To start with,
  we will show that by symmetry we need to consider only the first case $\alpha_t\le\pi/2$.

  Note that the intervention function $f$
  exhibits a symmetry
  $f(-u, v) = -f(u, v)$. Furthermore, we also have
  $\arccos\langle u,u'\rangle 
  = \pi-\arccos\langle u,-u'\rangle$. Consequently,
  \begin{align*}
      \alpha_{t+1}-\alpha_t
      &=\arccos\langle f(u,v),f(u',v)\rangle
      -\arccos\langle u,u'\rangle\\
      &=\pi-\arccos\langle f(u,v),f(-u',v)\rangle
      -(\pi-\arccos\langle u,-u'\rangle)
      \\
      &=-\big(\arccos\langle f(u,v),f(-u',v)\rangle
      -\arccos\langle u,-u'\rangle\big)\;.
  \end{align*}
  As a result, indeed it is enough that we prove~\eqref{eq:12}
  and then~\eqref{eq:13} follows by replacing
  $u'$ with $-u'$.
  
  Consider vector
  $v^*:=(\cos\eps,\sin\eps,0,\ldots,0)$
  (see Figure~\ref{fig:random-intervention}). We will now show that
  if $\eps\le\alpha_t\le\pi/2$ and the intervention $v$ is sufficiently close to $v^*$,
  then $v$ decreases the angle between $u$ and $u'$.
  To that end, let us use a metric on $S^{d-1}$ given by
  \begin{align*}
      D(u, v) := \arccos\langle u,v\rangle\;.
  \end{align*}
  Note that this metric is strongly equivalent to the standard Euclidean metric restricted to $S^{d-1}$. We can now use the triangle inequality to
  write
  \begin{align}\label{eq:10}
      \alpha_{t+1}&=D(f(u,v),f(u',v))\nonumber\\
      &\le D(f(u,v),f(u,v^*))+D(f(u,v^*),v^*)
      + D(v^*,f(u',v^*))+D(f(u',v^*),f(u',v))\;.
  \end{align}
  Let us bound the terms in~\eqref{eq:10} one by one.

  First, since, by~\eqref{eq:main}, $f(u,v^*)$ is a strict
  convex combination of $u$ and $v^*$
  (note that in our coordinate system 
  neither $u$ nor $v^*$ depends on $\alpha_t$), we have
  \begin{align*}
    D(f(u,v^*),v^*) = d(\eps) < D(u,v^*)=\eps\;.
  \end{align*}
  Similarly, 
  \begin{align*}
      D(v^*,f(u',v^*))\le D(v^*,u')=\alpha_t-\eps\;.
  \end{align*}
  Second, since $f$ is continuous,
  $D(v,v^*)<\delta'$ for small enough $\delta'>0$
  implies that both $D(f(u,v),f(u,v^*))$
  and $D(f(u',v^*),f(u',v))$
  are as small as needed
  (for example, less than $(\eps-d(\eps))/4$).
  
  All in all, we have that for some $\delta'=\delta'(\eps)>0$,
  \begin{align*}
      D(v,v^*)<\delta'\implies
      \alpha_{t+1}&<\frac{\eps-d(\eps)}{4}+ d(\eps)+(\alpha_t-\eps)+\frac{\eps-d(\eps)}{4}\\
      &=\alpha_t-\frac{\eps-d(\eps)}{2} \; .
  \end{align*}
  However, clearly, the event $D(v,v^*)<\delta'$
  has some positive probability $\delta''$.
  Therefore, taking 
  $\delta:=\min(\delta''/2,(\eps-d(\eps))/2)$,
  we have
  \begin{align*}
      \Pr\left[\alpha_{t+1}<\alpha_t-\delta \mid\alpha_t\right]>\delta\;,
  \end{align*}
  as claimed in~\eqref{eq:12}.
  \hfill\qedsymbol

\section{Asymptotic scenario: finding densest hemisphere}
\label{sec:densest-hemisphere-intro}

In this section
we study the asymptotic scenario
with one influencer who wishes to propagate
a campaign agenda $v^*\in\mathbb{R}^d$. We assume that the influencer
can use an unlimited number of interventions and its objective is to
make the opinions of as many agents as possible to converge to $v^*$.
More specifically, in this section we denote the initial opinions of
agents at time $t=1$ by $u_1,\ldots,u_n$.
Given these preexisting opinions of $n$ agents,
we want to find a sequence of interventions, $v^{(1)},v^{(2)},v^{(3)}\ldots$ that
maximizes the number of agents whose opinions converge to $v^*$.

The thrust of our results is that finding a good strategy
for the influencer is computationally hard. However,
both the optimal strategy and some natural heuristics
result in the polarization of agents.

\subsection{Equivalence of optimal strategy to finding densest hemisphere}

We first argue that the problem of finding an optimal
strategy is equivalent to identifying
an open hemisphere that contains the maximum number of agents.
An \emph{(open) hemisphere} is an intersection of the unit
sphere with a \emph{homogeneous open halfspace}
of the form 
$\left\{x\in\mathbb{R}^d:\langle x,a\rangle>0\right\}$
for some $a\in\mathbb{R}^d$.

\begin{theorem}\label{thm:maximum-halfspace}
  For any $v^*$, there exists a strategy to make at least $k$ agents converge to $v^*$
  if and only if there exists an open hemisphere containing at least $k$ of the
  opinions $u_1,\ldots,u_n$.
\end{theorem}

A surprising aspect of Theorem~\ref{thm:maximum-halfspace} is that
the maximum number of agents that can be persuaded does not depend
on the target vector $v^*$. As we argue in Remark~\ref{rem:moving-around},
this is somewhat plausible in the long-term setting with unlimited
number of interventions. We also note that the number of interventions required to bring the opinions up to a given level of closeness to $v^*$ \emph{does} depend
on $v^*$.


\begin{proof}[Proof of Theorem~\ref{thm:maximum-halfspace}]
First, we prove that the hemisphere condition is sufficient for the existence of a strategy to make the agents' opinions converge (Claim~\ref{claim:lower-bd}). Then we prove the trickier direction: that the hemisphere condition
is also \emph{necessary} for the existence of such a strategy (Claim~\ref{claim:upper-bd}).

\begin{claim}
  \label{claim:lower-bd}
If opinions $u_1,\ldots,u_k$ are contained in an open 
hemisphere,  then there is a sequence of interventions
making all of $u_1,\ldots,u_k$ converge to $v^*$.
\end{claim}
\begin{proof}
  By definition of open hemisphere, 
  there is a vector $a \in \R^d$
  such that $\langle u_i,a \rangle > 0$ for every agent $i = 1,\ldots,k$.
  By~\eqref{eq:main}, it is clear that repeated application 
  of $a$ makes all the points converge to $a$ 
  as time $t\to\infty$ .

  After all the points are clustered close enough to $a$, by a similar argument
  they can be ``moved around'' together towards another arbitrary point $v^*$.
  For example, if $\langle v^*, a\rangle>0$, the intervention $v^*$ can be applied
  repeatedly. If $\langle v^*, a\rangle\le 0$, one can proceed in two stages,
  first applying an intervention proportional to $(v^*+a)/2$, and then applying $v^*$.
\end{proof}

\begin{remark}\label{rem:moving-around}
As a possible interpretation of the mechanism in 
Claim~\ref{claim:lower-bd}, it is not unheard of in campaigns
on political issues to use an analogous strategy. First,
build a consensus around a 
(presumably compromise) opinion. Then,
``nudge'' it little by little towards another direction.

In an extreme case one can imagine this mechanism even
flipping the opinions of two polarized clusters. One example
of this could be the reversal of the opinions on certain issues of 20th century
Republican and Democratic parties in the US (this particular phenomenon can be found in many texts, e.g.~\cite{Kuziemko18}).
\end{remark}
\smallskip

To prove the other direction of Theorem \ref{thm:maximum-halfspace}, we will rely on the notions of conical combination and convex cone.
A \emph{conical combination} of points $u_1,\ldots,u_n\in\R^d$ is any point
of the form $\sum_{i=1}^n \alpha_iu_i$ where $\alpha_i\ge 0$ for every $i$.
A \emph{convex cone} is a subset of $\R^d$ that is closed under
finite conical combinations of its elements. Given a finite set of points
$S\subseteq\R^d$, the convex cone \emph{generated} by $S$ is the smallest convex cone
that contains $S$.

\begin{claim}
\label{claim:convex}
Suppose that for a given sequence of interventions, the opinions
$u_1, \ldots, u_n$ converge to the same point $v^*$.
Then, for any unit vector $u_{n+1}$ that lies in the convex cone
of $u_1,\ldots,u_n$, we have that $u_{n+1}$ also converges to $v^*$.
\end{claim}
\begin{proof}
  It suffices to prove that if at time $t$ an opinion
  $u_{n+1}^{(t)}$ lies in the convex cone of
  other opinions
  $u_1^{(t)},\ldots,u_n^{(t)}$, then after applying one intervention 
  $v^{(t)}$ the new
  opinion $u_{n+1}^{(t+1)}$ lies in the convex cone of $u_1^{(t+1)},\ldots,u_n^{(t+1)}$.
  Then the claim follows by induction.

  To prove this, we can simply write out $u_{n+1}^{(t+1)}$, using the relation
  $u_{n+1}^{(t)}=\sum_{i=1}^n \lambda_i u_i^{(t)}$ (where we use the notation $u \propto v$
  to mean that $u = c\cdot v$ for some constant $c>0$):
\begin{align}
  u_{n+1}^{(t+1)}
  &\propto  u_{n+1}^{(t)} + \eta \left\langle u_{n+1}^{(t)}, v^{(t)} \right\rangle\cdot v^{(t)}  \nonumber \\
  &= \sum\limits_{i=1}^{n} \lambda_i u_i^{(t)} + \eta \cdot \sum\limits_{i=1}^{n} \lambda_i \left\langle u_i^{(t)}, v^{(t)} \right\rangle\cdot v^{(t)}  \nonumber \\
                      &=\sum\limits_{i=1}^{n} \lambda_i 
                      \left( u_i^{(t)} + \eta \cdot \left\langle  u_i^{(t)}, v^{(t)}\right\rangle\cdot v^{(t)} \right) \nonumber \\
    &= \sum\limits_{i=1}^{n} \lambda_i \cdot c_i u_i^{(t+1)} \label{eq:convex-1}
\end{align}

where the constants in~\eqref{eq:convex-1} are
$c_i := \left\|u_i^{(t)} + \eta \cdot \langle  u_i^{(t)}, v^{(t)}\rangle \cdot v^{(t)} \right\|$. 
Specifically, they are all nonnegative.
\end{proof}

\begin{claim}
\label{claim:antipodal}
Suppose there are two opinions $u_1, u_2$ that are antipodal, i.e.,
$u_1 = - u_2$. Then these two opinions will remain antipodal in
future time steps. In particular, they will never converge to a single point.
\end{claim}
\begin{proof}
  This follows directly from~\eqref{eq:main}, noting that, for any intervention
  $v$, we have 
  $u_1 + \eta \cdot \left\langle  u_1, v\right\rangle \cdot v =\allowbreak -\left(u_2 + \eta \cdot \left\langle  u_2, v\right\rangle  \cdot v \right)$.
  \hfill\qedsymbol
\end{proof}
We will also use the following consequence of the separating hyperplane theorem:
\begin{claim}
\label{fact:convex}
A collection of unit vectors $a_1, \ldots, a_n$ cannot be placed in an open hemisphere
if and only if the zero vector lies in the convex hull of
$a_1, \ldots, a_n$.
\end{claim}
Now we are ready to establish the reverse implication
in Theorem~\ref{thm:maximum-halfspace}.

\begin{claim}
\label{claim:upper-bd}
Suppose that we start with agent opinions $u_1, \ldots, u_n$ and that there is
no hemisphere that contains $M$ of those opinions.
Then, there is no strategy that makes $M$ of the opinions converge to the same
point.
\end{claim}

\begin{proof}
  Assume towards contradiction that there exists a strategy that makes $M$ opinions converge to the same point,
  and assume wlog~that they are $u_{1}, \ldots, u_{M}$. By assumption, we know that
  there is no hemisphere that contains all of $u_1,\ldots,u_M$, hence, by Claim~\ref{fact:convex}, there is a convex combination of
  $u_1,\ldots,u_M$ that equals 0. Therefore, there is also a conical combination
  of $u_{1}, \ldots, u_{M-1}$ that equals $- u_{M}$, where wlog we assume
  that the coefficient on $u_{M}$ is initially nonzero. 
  By Claim~\ref{claim:convex}, we conclude that if $u_{1}, \ldots, u_{M-1}$ converge to the same point, then so does $- u_{M}$. But that means that $-u_M$ and $u_M$ converge to
  the same point, which is a contradiction by Claim~\ref{claim:antipodal}.
\end{proof}
That concludes the proof of Theorem~\ref{thm:maximum-halfspace}.
\end{proof}

\begin{remark}
  One consequence of Theorem~\ref{thm:maximum-halfspace} is that if the
  agent opinions are initially distributed uniformly on the unit sphere, and if the number of agents $n$ is large compared
  to the dimension $d$,
  an optimal strategy converging as many opinions as possible to $v^*$ results, with high probability, in dividing
  the population into two groups of roughly equal size, where
  the opinions inside each group converge to one of two antipodal
  limit opinions (i.e., $v^*$ and $-v^*$).
  Furthermore, this optimal strategy, which, as discussed below,
  might not be easy to implement, will not perform significantly
  better than a very simple strategy of fixing a random intervention
  and applying it repeatedly. Of course the simple strategy will also
  polarize the agents into two approximately equally large groups. 
\end{remark}
\smallskip

\subsection{Computational equivalence to learning halfspaces}
\label{sec:halfspace-reduction}

Theorem~\ref{thm:maximum-halfspace} implies
that an optimal strategy for the influencer 
is to compute the open hemisphere that is the densest,
i.e., it contains the most opinions,
and then apply the procedure from Claim~\ref{claim:lower-bd}
to converge the opinions from this hemisphere to $v^*$.
In this section we study the computational complexity
of this problem. While different approaches are possible,
we focus on hardness of approximation and worst-case complexity.

\begin{definition}[Densest hemisphere]
    The input to the \emph{densest hemisphere} problem
    consists of parameters $n$ and $d$
    and a set of $n$ unit vectors
    $D=\{u_1,\ldots,u_n\}$ with $u_i\in S^{d-1}$.
    The objective is to find vector $a\in S^{d-1}$
    maximizing the number of points from $D$ 
    that belong to the open halfspace
    $\left\{x\in\mathbb{R}^d: \langle x,a\rangle>0\right\}$.
\end{definition}

We analyze the computational complexity of the densest hemisphere
problem in terms of the number of vectors $n$, regardless of
dimension $d$. In particular, the computationally hard instances
that exist as we will show in Theorem~\ref{thm:hardness-of-approximation} 
have high dimension, without any guarantees
beyond $d\le n$ (which can always be assumed wlog).
On the other hand, the algorithms from 
Theorem~\ref{thm:stable-hemisphere-algo} run in time polynomial
in $n$ uniformly for all values $d\le n$.

In contrast, the case of finding densest hemisphere in fixed
dimension $d$ can be solved efficiently.
For example, an optimal solution can be found by considering $O(n^d)$ halfspaces defined by $d$-tuples of input vectors.
We omit further details.

Our main result in this section relies on equivalence
of the densest hemisphere problem 
and the problem of \emph{learning noisy halfspaces}.
Applying a work by Guruswami and Raghavendra~\cite{GR09} we 
will show that it is computationally difficult to even approximate
the densest hemisphere up to any non-trivial constant factor:

\begin{theorem}\label{thm:hardness-of-approximation}
  Unless P=NP, for any $\eps>0$, there is no polynomial time algorithm $A_\eps$
  that distinguishes between instances of densest hemisphere problem
  such that, letting $D:=\{u_1,\ldots,u_n\}$:
  \begin{itemize}
  \item Either there exists a hemisphere $H$ such that $|D\cap H|/n > 1-\eps$.
  \item Or for every hemisphere $H$ we have $|D\cap H|/n <1/2+\eps$.
  \end{itemize}
  Consequently, unless P=NP, for any $\eps>0$ there is no polynomial time algorithm $A_\eps$ that,
  given an instance $D$ that has a hemisphere with density more than $1-\eps$,
  always outputs a hemisphere with density more than $1/2+\eps$.
\end{theorem}

In other words, even if guaranteed the existence of an
extremely dense hemisphere, no polynomial time algorithm
can do significantly better than choosing an arbitrary
hyperplane and outputting the one of its two hemispheres
that contains the larger number of points.
At the same time, \cite{BDS00} (relying on earlier work \cite{BDES02})
shows that there exists an algorithm that finds a dense hemisphere provided
that this hemisphere is stable in the sense that it remains dense even after a small
perturbation of its separating hyperplane:
\begin{theorem}[\cite{BDS00}]
\label{thm:stable-hemisphere-algo}
  For every $\mu > 0$, there exists a polynomial time algorithm $A_\mu$, that,
  given an instance $D=\{u_1,\ldots,u_n\}$ of the densest hemisphere problem, provides the
  following guarantee:

  Let $a\in S^{d-1}$ be the vector that maximizes the size
  of intersection $|D\cap H_{a,\mu}|$ for halfspace
  $H_{a,\mu} = \{x: \langle x,a\rangle>\mu\}$.
  Then, the algorithm $A_\mu$ outputs
  a hemisphere corresponding to a homogeneous halfspace 
  $H_{a'}=\{x:\langle x,a'\rangle>0\}$ such that 
  $|D\cap H_{a'}|\ge|D\cap H_{a,\mu}|$.
\end{theorem}

We emphasize that the only inputs to the algorithms
are $n$, $d$ and the set of vectors $D$,
and that the complexity is measured as a function of $n$.
For example, the algorithm $A_\mu$
runs in polynomial time for every $\mu>0$, but the running
time is not uniformly polynomial in $1/\mu$.

\medskip

In the remainder of this section we elaborate on how to obtain
Theorem~\ref{thm:hardness-of-approximation} from known results.
To that end, we start with defining the related problem
of finding maximum agreement halfspace.

\begin{definition}[Maximum Agreement Halfspace]
  In the problem of \emph{maximum agreement halfspace}, 
  the inputs are parameters $n$ and $d$, 
  and a labeled set of points
  $D = \{(x_1,y_1),\ldots,(x_n,y_n)\}\in\mathbb{R}^d\times\{\pm 1\}$.
  The objective
  is to find a halfspace $H = \{x: \langle x,a\rangle > c\}$
  for some $a\in\mathbb{R}^d$ and $c\in\mathbb{R}$ which
  maximizes the
  \emph{agreement}
  \begin{align*}
    A(D, H) = \frac{\sum_{i=1}^n \mathbbm{1}\left[
    y_i\cdot x_i \in H
    \right]}{n}\;.
  \end{align*}
\end{definition}

There is a strong hardness of approximation
result for maximum halfspace agreement
\cite{GR09} (see also \cite{FGK06, BB06, BEL03, AK98} for related work):
\begin{theorem}[\cite{GR09}]\label{thm:gr}
Unless P=NP, for any $\eps>0$, there is no polynomial time
  algorithm $A_\eps$ that distinguishes the following
  cases of instances of maximum agreement halfspace
  problem:
  \begin{itemize}
  \item There exists a halfspace $H$ such that $A(D, H)>1-\eps$.
  \item For every halfspace $H$ we have $A(D, H)<1/2+\eps$.
  \end{itemize}
\end{theorem}

As in Theorem~\ref{thm:hardness-of-approximation}, 
the hard instances are not guaranteed to have any dimension bounds beyond
trivial $d\le n$.

As pointed out in~\cite{BDS00}, there exists a reduction from the
maximum agreement halfspace problem to the densest hemisphere problem that
preserves the quality of solutions. Since this reduction is only briefly
sketched in~\cite{BDS00}, we describe it below.

The reduction proceeds as follows: Given a labeled set
$D = \{(x_1,y_1),\ldots,(x_n,y_n)\}\in\mathbb{R}^d\times \{\pm 1\}$, we map it
to $D' =\{x'_1,\ldots,x'_n\}\in\mathbb{R}^{d+1}$ using the formula
\begin{align*}
  x'_i = \frac{1}{\sqrt{1+\|x_i\|^2}}\cdot (y_i x_i, 1)\;. 
\end{align*}
In other words, we proceed in three steps: 
first, we negate
each point that came with negative label $y_i=-1$.
Then, we add a 
new coordinate and set its
value to $1$ for every point $x_i$. Finally, we normalize each resulting point
so that it lies on the unit sphere in $S^{d}$. 

This is a so-called ``strict reduction'', which is expressed in the
following claim:
\begin{claim}\label{cl:hemisphere-reduction}
  The solutions (halfspaces) for an instance $D$
  of Maximum Agreement Halfspace are in one-to-one correspondence
  with solutions (hemispheres) for the reduced instance
  $D'$ of Densest Hemisphere.
  Furthermore,
  for a corresponding pair of solutions $(H, H')$ the agreement $A(D,H)$ is
  equal to the density $|D'\cap H'|/n$.
\end{claim}

\begin{proof}
  It is more convenient to think of solutions for $D'$ as homogeneous, open
  halfspaces $H' = \{x\in\mathbb{R}^{d+1}: \langle x,a\rangle > 0\}$.

  With that in mind, we map a solution to the maximum agreement halfspace
  problem $H=\{x\in\mathbb{R}^d:\langle x,a\rangle > c\}$ to a solution to
  the densest hemisphere problem $H'=\{(x,x_{d+1})\in\mathbb{R}^{d+1}:
  \langle (x,x_{d+1}),(a,-c))\rangle >0\}$. Clearly, this is a one-to-one mapping
  between open halfspaces in $\mathbb{R}^d$ and homogeneous open halfspaces
  in $\mathbb{R}^{d+1}$.

  Furthermore, it is easy to verify that $y_i\cdot x_i\in H$ if and only if
  $x'_i \in H'$ and therefore $A(D, H) = |D'\cap H'|/n$.
\end{proof}

Theorem~\ref{thm:hardness-of-approximation} follows
from Theorem~\ref{thm:gr} and
Claim~\ref{cl:hemisphere-reduction}
by standard (and straightforward) arguments
from complexity theory.

\section{Short-term scenario: polarization as externality}
\label{sec:short-term intro}
The analysis of the asymptotic setting with unlimited interventions tells us what is feasible and what is not. A fundamentally different question is how to persuade as many as possible with a limited number of interventions. This is motivated by bounded resources or time that usually allow
only limited placements of campaigns and advertisements.
Furthermore, arguably only the initial interventions can be considered effective:
in the long run the opinions might shift due to external factors
and become more unpredictable and harder to control. Therefore, in this section we discuss
strategies where the influencer has only one intervention at its
disposal, and its goal is to get as many agents as possible to exceed certain
``threshold of agreement'' with its preferred opinion. 
Throughout this section, we fix $\eta=1$ in Equation~\ref{eq:07}, so an opinion $u$ is updated
to be proportional to $w = u + \langle u,v\rangle\cdot v$.

Both scenarios we discuss in this section describe a situation where
a ``new'' product or idea is introduced. Therefore, we assume
that the agents have some preexisting opinions in $\mathbb{R}^{d-1}$
and that they are neutral as to the new idea, with the $d$-th coordinate
set to zero for every agent. Our results indicate significant potential for
polarization in such a situation. This is in spite of the fact that
the influencer might only care about persuading a number of agents
towards the new subject, without intention to polarize. 

Since we are dealing with scenarios with only one intervention,
we use the following notational convention: an initial
opinion of agent $i$ is denoted $u_i$ and the opinion after
intervention is denoted $\uafter_i$.

\subsection{One intervention, two agents: polarization costs}
\label{subsec:short-term-1}

We consider a simple example that features only two agents 
and one influencer who is allowed one intervention.
We imagine a new product, such that
the agents are initially agnostic about it, i.e., $u_{i,d}=0$ for
$i=1,2$. Given an intervention $v$, we are interested in two issues:
First, what will be new opinions of agents about the product $\uafter_{i,d}$?
Second, assuming that the initial correlation between opinions is
$c=\langle u_1,u_2\rangle$, what will be the new correlation
$\cafter=\langle \uafter_1,\uafter_2\rangle$? We think of the correlation as a measure of
agreement between the agents and therefore interpret differences in correlation
as changes in the extent of polarization.

In order to answer these questions, we introduce notions of two- and
one-agent interventions corresponding to two natural
strategies:
\begin{definition}
    The \emph{two-agent intervention} is an intervention
    that maximizes $\min(\uafter_{1,d},\uafter_{2,d})$. The
    \emph{one-agent intervention} maximizes
    $\max(\uafter_{1,d}, \uafter_{2,d})$.
\end{definition}

The motivation for this definition is as follows. Assume that there exists a threshold $T>0$
such that agent $i$ is going to make a positive decision (e.g., buy the product or vote
a certain way) if its coordinate $\uafter_{i,d}$ exceeds $T$.
Then, if the influencer cares only about inducing agents to make the decision, it
has two natural choices for the intervention.
One option is the case where it is possible to induce two decisions, i.e.,
achieve $\uafter_{1,d},\uafter_{2,d}>T$. By continuity considerations, it is not difficult to see
that an intervention that achieves this can be assumed to maximize $\min(\uafter_{1,d},\uafter_{2,d})$
with $\uafter_{1,d}=\uafter_{2,d}$
(such intervention is also optimal if the influencer bets on convincing both agents without
knowing $T$).
The other case is to appeal only to one of the agents, disregarding
the second agent and concentrating only on achieving, say, $\uafter_{1,d}>T$.

Let $c=\langle u_1,u_2\rangle$ be the initial correlation
between opinions and let $\ctwo$ and 
$\cone$ be the correlations
after applying, respectively, the two- and one-agent
interventions. Our main result in this section is:
\begin{proposition}\label{prop:polarization-cost}
Let $\rho:=\ctwo-\cone$ be a value that we call the
\emph{polarization cost}. Then, we always have $\rho\ge 0$
with exact values given as
\begin{align}\label{eq:25}
\ctwo=1-\frac{\sqrt{2}(1-c)}{\sqrt{3c+5}}\;,
\qquad\qquad
\cone=\frac{c\sqrt{2}}{\sqrt{c^2+1}}\;.
\end{align}
\end{proposition}

The values of $\rho,\ctwo$ and $\cone$ as functions
of $c$ are illustrated in Figure~\ref{fig:tradeoff-example-left}.
Proposition~\ref{prop:polarization-cost} states that the
one-agent intervention always results in smaller correlation
than the two-agent intervention. 
Note that we made a modeling assumption that the influencer will
always choose an intervention as opposed to doing nothing. This is
consistent with a scenario where the influencer's objective
is to increase the opinions above the threshold
$T$. In that case doing nothing is certain
to give no gain to the influencer.

The main conclusion of this theorem is consistent with our other results.
In the setting we consider,
in the absence of any external mitigation, the self-interested influencer without
direct intention to polarize might be 
incentivized to choose the intervention that increases polarization.
If polarization is regarded as undesirable, the polarization cost can be thought of
as the externality imposed on the society.

Looking at Figures~\ref{fig:tradeoff-example-left} and~\ref{fig:tradeoff-example-right}, this effect seems most pronounced
for initial correlation around $c\approx-0.5$, where the one-agent intervention increases polarization,
the polarization cost is large and
the range of thresholds $T$ for which the influencer profits from the one-agent
strategy is relatively large. This suggests that a situation
where the society is already somewhat polarized is particularly vulnerable
to spiraling out of control.  It also suggests that situations where the level of commitment
required for the decision (i.e., the threshold $T$) is large increase the risk of polarization.

We also note that this overall picture is complicated by the case of positive initial
correlation $c>0$. In that case both two- and one-agent interventions actually increase
the correlation between the agents, even though the two-agent intervention does so to a larger extent. The analysis leading to the proof of Proposition~\ref{prop:polarization-cost} is contained in 
Appendix~\ref{app:polarization-cost}.

\begin{figure}[!ht]
    \centering
    \begin{minipage}{0.45\textwidth}
    \centering
    \includegraphics[width=0.9\textwidth]{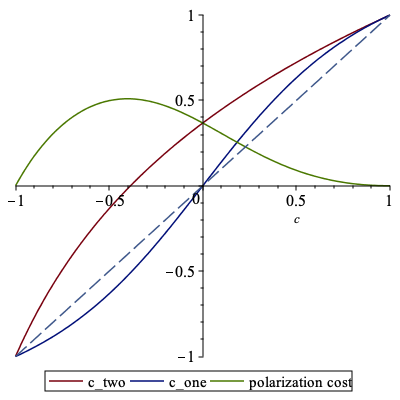} 
    \caption{Illustration of the polarization cost as a function of the initial
    correlation $c$. The dashed line is the initial 
    correlation included
    as a reference point.
    The red and blue lines are correlations after applying
    two- and one-agent interventions respectively.
    The green line shows the polarization cost $\ctwo-\cone$.}
    \label{fig:tradeoff-example-left}
\end{minipage}\hfill
    \begin{minipage}{0.45\textwidth}
    \centering
    \includegraphics[width=0.9\textwidth]{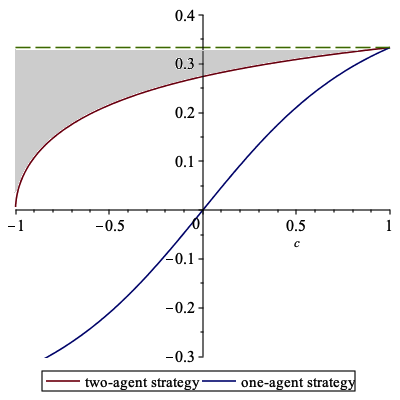} 
    \caption{The after-intervention opinions of both
    agents $\uafter_{i,d}$ as functions of initial correlation
    $c$.
    The red line represents the opinion of either agent 
    after applying
    the two-agent intervention. The blue line is the opinion
    of the second agent after the one-agent intervention. For reference, the dashed
    line ($1/3$) shows the opinion of the first agent in the one-agent intervention (which does not depend on $c$).
    The grey area represents the range of thresholds $T$ where it is preferable for the influencer to
    apply the one-agent intervention.}
    \label{fig:tradeoff-example-right}
    \end{minipage}
\end{figure}

\subsection{One intervention, many agents: finding the densest spherical cap}
\label{subsec:short-term-2}

A more general version of the problem of persuading with limited number of interventions features
$n$ agents with opinions $u_1, \ldots, u_n \in \mathbb{R}^d$.
The influencer is given a threshold $0\le T< 1$
and can apply one intervention $v$ with the objective of maximizing the number of agents
such that $\uafter_{i,d} > T$. As before, 
we assume that intially $u_{i,d}=0$ and that $T$ can be interpreted as a
threshold above which a consumer decides to buy the newly advertised product, 
or more generally take a desired action, such as voting, donating, etc. 

Interestingly, we show that this problem
is equivalent to a generalization of the densest hemisphere problem
from the long-term scenario discussed in~Section~\ref{sec:densest-hemisphere-intro}.
More precisely, it is equivalent to finding a densest \emph{spherical cap}
of a given radius (that depends on the threshold $T$) in $d-1$ dimensions.

We give the technical statement in the proposition below. 
We make an assumption $0\le T<1/3$, since $1/3$ is the maximum
value that can be achieved in the $d$-th coordinate
by a single intervention, cf.~Figure~\ref{fig:tradeoff-example-right}.
In order to state Proposition~\ref{prop:short-term-cap}, we slightly 
abuse notation
and write vectors $u\in\mathbb{R}^d$ as
$u=(u^*, u_d)$ for $u^*\in\mathbb{R}^{d-1}$, $u_d\in\mathbb{R}$.

\begin{proposition}\label{prop:short-term-cap}
In the setting above, let
\begin{align*}
    c:=\frac{2T}{1-3T^2}\;,\qquad
    z:=\frac{\sqrt{\sqrt{1+3c^2}-1}}{\sqrt{3}c}\;,\qquad
    \beta:=\arccos(z)\;.
\end{align*}
Then, the number of agents with $\uafter_{i,d}> T$ is
maximized by applying an intervention
\begin{align}\label{eq:08}
    v := (\cos\beta\cdot v^*,\sin\beta)
\end{align}
for a unit vector $v^*\in\mathbb{R}^{d-1}$ that maximizes
the number of agents satisfying
\[
    \langle u_i^*, v^*\rangle> c\;.
\]
\end{proposition}

The proof of Proposition~\ref{prop:short-term-cap} is contained in
Appendix~\ref{app:short-term-cap}.
Note that the solution to this short-term problem for $T$ going to zero approaches the densest hemisphere solution
to the long-term problem discussed in
Section~\ref{sec:densest-hemisphere-intro}.

\section{Asymptotic effects of two dueling influencers: two randomized interventions polarize}
\label{sec:two-advertisers-intro}

Finally, we analyze a scenario where there are two influencers
with differing agendas, represented by different\footnote{We also assume that $v \neq -v'$, as otherwise the intervention effects are the same in our model.}
intervention vectors $v$ and $v'$.
We consider the \emph{randomized} setup, where at each time step, one of the
influencers is randomly chosen
to apply their intervention. We demonstrate that this setting also results, in
most cases and in a certain sense, in the polarization of agents. 

Recall that a convex cone of two vectors $v$ and $v'$
is the set $\{\alpha v + \beta v': \alpha, \beta\ge 0\}$.
A precise statement that we prove is:
\begin{theorem}
\label{thm:two-advertisers}
Let $\langle v,v'\rangle > 0$ and let a starting opinion $u^{(1)}$ be such that
$\langle u^{(1)}, v\rangle \ne 0$ or $\langle u^{(1)}, v'\rangle \ne 0$. Then,
as $t$ goes to infinity and almost surely,
either the Euclidean distance between $u^{(t)}$ and
the convex cone generated by $v$ and $v'$ or between 
$u^{(t)}$ and
the convex cone generated by $-v$ and $-v'$ goes to 0.
\end{theorem}

In order to justify the assumptions of Theorem~\ref{thm:two-advertisers},
note that if an agent starts with
an opinion $u^{(1)}$ such that
\begin{align}\label{eq:04}
  \langle u^{(1)},v \rangle = \langle u^{(1)},v'\rangle=0\;,
\end{align}
applying $v$ or $v'$ never changes their opinion.
In Theorem~\ref{thm:two-advertisers} we show that if~\eqref{eq:04} does
not hold and,
additionally, $\left< v, v'\right> > 0$,
(if $\langle v,v'\rangle<0$ we can exchange $v'$ with
$-v'$ without changing the effects of any interventions),
the opinion vector with probability 1
ends up either converging to the convex cone generated by $v$ and $v'$
or the convex cone generated by $-v$ and $-v'$.
In particular, since vectors $u$ for which~\eqref{eq:04} holds form a
set of measure 0, if $n$ initial opinions are sampled iid from an absolutely
continuous distribution, almost surely all opinions converge to the convex cones (which are themselves sets of measure 0 for $d>2$).

Furthermore, this notion of polarization is strengthened if the correlation between the two interventions is large.
As in Theorem~\ref{thm:random-polarization}, the best we can hope for is that for each pair of opinions either the distance
between $u_1^{(t)}$ and $u_2^{(t)}$ or between 
$u_1^{(t)}$ and $-u_2^{(t)}$ converges to 0.
Letting $V:=\vecspan\{v,v'\}$ and $W:=V^{\perp}$ and
writing any vector $u$ as a sum of its respective projections
$u = u_V + u_W$, we show:

\begin{theorem}\label{thm:contraction}
  Suppose that $\langle v,v'\rangle > 1/\sqrt{2+\eta}$ and let $u_1^{(1)}, u_2^{(1)}$ be such that
  $(u_1^{(1)})_V\ne 0$, $(u_2^{(1)})_V\ne 0$. Then, almost surely,
  either $\|u_1^{(t)}-u_2^{(t)}\|$ converges to 0, or
  $\|u_1^{(t)}+u_2^{(t)}\|$ converges to 0.
\end{theorem}

In other words, the stronger notion of convergence, same as in Section \ref{sec:random-polarization} with uniformly drawn random interventions, reappears in case the correlation between two interventions $v$ and $v'$
is larger than $1/\sqrt{2+\eta}$.
In particular, we have strong convergence
for any $\eta>0$ and 
$\langle v,v'\rangle\ge\sqrt{2}/2\approx 0.71$,
and for $\eta=1$ for
$\langle v,v'\rangle>\sqrt{3}/3\approx 0.58$.
Our experiments suggest that this convergence occurs also for other non-zero values
of the correlation $\langle v,v'\rangle$, but we do not prove it here.

Also note that same in spirit as Remark 3.1, the usual argument from symmetry shows that if the
initial opinions are independent samples from a symmetric
distribution, then with high probability the opinions divide
into two clusters of roughly equal size.

\medskip

The case when $v$ and $v'$ are orthogonal is
different. As we mentioned, if $\langle v,v'\rangle>0$, i.e.,
the angle between $v$ and $v'$ is less than $\pi/2$,
then all opinions converge to the two ``narrow'' convex cones,
respectively between $v$ and $v'$ and between $-v$ and $-v'$ 
--- namely, the pairs of vectors among $v, v', -v$, and $-v'$ between which there are acute angles.
Similarly, if $\langle v,v'\rangle<0$, 
then the opinions converge to two cones
between $v$ and $-v'$ and between $-v$ and $v'$.
In case $\langle v,v'\rangle=0$ the four convex cones
form right angles, so such a result is not possible.

However, we can still show that
an initial opinion $u^{(1)}$ converges to the same quadrant in which it starts
with respect to $v$ and $v'$. 
Namely, for all $t$, we have that
$\sgn\left(\left< u^{(t)}, v\right>\right) = \sgn\left( \left\langle u^{(1)}, v\right\rangle \right)$ and $\sgn\left(\left\langle u^{(t)}, v'\right\rangle \right) = \sgn\left(\left\langle u^{(1)}, v'\right\rangle\right)$, and furthermore the distance between
$u^{(t)}$ and the subspace $V$ goes to 0 with $t$:

\begin{proposition}
\label{cor:orthogonal-advertisers}
Suppose that $\langle v, v'\rangle = 0$ and let an initial opinion $u^{(1)}$ be such
that $\langle u^{(1)}, v\rangle\ne 0$ \emph{and}
$\langle u^{(1)}, v'\rangle\ne 0$. Then, almost surely, the following facts hold:
\begin{enumerate}
\item $\|u^{(t)}_{W}\| \to 0$ as $t \to \infty$.
\item For all $t$, $\sgn\left(\left< u^{(t)}, v\right>\right) = \sgn\left( \left\langle u^{(1)}, v\right\rangle \right)$ and $\sgn\left(\left\langle u^{(t)}, v'\right\rangle \right) = \sgn\left(\left\langle u^{(1)}, v'\right\rangle\right)$.
\end{enumerate}
\end{proposition}

Fascinatingly, Gaitonde, Kleinberg and Tardos~\cite{GKT21} 
showed subsequently to our initial preprint
that strong polarization does not occur for orthogonal interventions.
Specifically, they proved that two opinions in $S^{d-1}$ with
random interventions chosen iid from the standard basis
$\{e_1,\ldots,e_d\}$ do not polarize in the sense of
$u_1^{(t)}-u_2^{(t)}$ or $u_1^{(t)}-u_2^{(t)}$ vanishing,
but they do exhibit a weaker form of polarization.
We refer to their paper for more details.

\medskip

In order to prove Theorem~\ref{thm:two-advertisers},
we first show that the distance between $u^{(t)}$ and $V$ almost surely goes to 0 as
$t \to \infty$, by showing that the norm of the projection of $u^{(t)}$
onto $W$ converges to 0.
Then, we demonstrate that the convex cone spanned by $v$ and $v'$ is absorbing:
when the projection of $u^{(T)}$ onto $V$ falls in the cone,
then the projections of $u^{(t)}$ for $t \geq T$ always stay in the cone
as well. 

Finally, 
we show that 
almost surely
the projection of $u^{(t)}$ onto $V$ eventually enters either
the cone spanned by $v$ and $v'$, or the cone spanned by $-v$ and $-v'$.
More concretely, we show that at any time $t$, there is a sequence of $T$
interventions that lands the projection of $u^{(t+T)}$ in one of the cones,
for some $T$ that is independent of $t$.
Since this sequence occurs with probability $2^{-T}$, which is 
independent of $t$, the opinion almost surely eventually enters one of the cones.


\subsection{Proofs of Theorem~\ref{thm:two-advertisers}
and Proposition~\ref{cor:orthogonal-advertisers}}

We start with the fact the opinions converge to the subspace $V$ spanned
by the two intervention vectors.
Recall that $V=\vecspan\{v,v'\}$ and that $W=V^{\perp}$. In the following we will
write $\langle v,v'\rangle=\cos\theta$
for $0<\theta\le\pi/2$.

\begin{proposition}
\label{prop:two-advertisers-subspace}
Let $\langle v,v'\rangle\ge 0$ and take an opinion vector $u$ such that
$\|u_V\| = c \ge 0$. Furthermore, let $\uafter$ be the vector resulting from
randomly intervening on $u$ with either $v$ or $v'$.
Then:
\begin{enumerate}
\item $\|\uafter_W\|^2 \leq \|u_W\|^2$.
\item With probability at least 1/2,
  $\|\uafter_{W}\|^2 \leq \|u_{W}\|^2 \cdot (1 - \xi)$,
  where
  \begin{align*}
  \xi = \min\left(\frac{1}{2},
  (\eta+\eta^2/2) \cdot \frac{c^2\theta^2}{16}
  \right)\;.
  \end{align*}
\end{enumerate} 
\end{proposition}
\begin{proof}
  Recall from~\eqref{eq:main}--\eqref{eq:main3}
  that if $\vbar\in\{v,v'\}$ is the intervention vector,
  then
\[ \uafter = k (u + \eta\left\langle u,\vbar\right\rangle\cdot \vbar)\]
where $k=\sqrt{\frac{1}{1+(2\eta + \eta^2)\cdot\langle u,\vbar\rangle^2}}$ is the normalizing constant.
Observe that when we project onto $W$, the component in the direction of $\vbar$ vanishes, so we have that
\[ \uafter_{W} = k \cdot u_{W}\;,\]
and the first claim easily follows since $k\le 1$.

To establish the second point, we need to show that with probability 1/2
we have $k^2 < 1$ or, equivalently,
$\langle u,\vbar\rangle^2=\langle u_V,\vbar\rangle^2 > 0$.
Since $\theta \ne 0$, the projected vector $u_V$ cannot be orthogonal both
to $v$ and $v'$ (cf.~Figure~\ref{fig:angles}).
More precisely, for at least one of $\vbar\in\{v,v'\}$ the primary angle between
$u_V$ and $\vbar$ (or $-\vbar$) must be at most $\pi/2-\theta/2$ and consequently
\[
  \left|\left\langle u_V,\vbar\right\rangle\right| \geq \|u_V\|
  \cdot |\cos(\pi/2 - \theta/2)| \geq c \cdot \theta / 4 \;,
\]
resulting in
\[
  k^2=
  \frac{1}{1+(2\eta + \eta^2)\cdot\langle u_V, \vbar\rangle^2}
  \le \max\left(\frac{1}{2},
  1-(\eta+\eta^2/2) \cdot 
  \frac{c^2\theta^2}{16}\right)\;.
  \qedhere
\]
\end{proof}

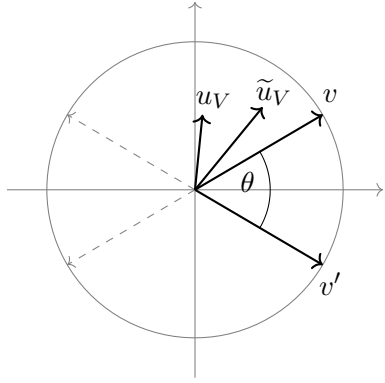
\begin{figure}[ht]\centering
  \begin{tikzpicture}
    \draw [->, color=gray] (-2.5, 0) -- (2.5, 0);
    \draw [->, color=gray] (0, -2.5) -- (0, 2.5);
    \draw [color=gray, domain=0:360, samples=100] plot ({1.97*cos(\x)}, {1.97*sin(\x)});
    
    \draw [thick, ->] (0, 0) -- (1.7, 1); 
    \draw [thick, ->] (0, 0) -- (1.7, -1);
    \draw [color=gray, dashed, ->] (0, 0) -- (-1.7, 1); 
    \draw [color=gray, dashed, ->] (0, 0) -- (-1.7, -1);
    \node at (1.8, 1.25) {$v$};
    \node at (1.8, -1.25) {$v'$};
    \draw [thick, ->] (0, 0) -- (0.1, 1);
    \node at (0.25, 1.2) {$u_V$};
    \draw [thick, ->] (0, 0) -- (0.9, 1.1);
    \node at (1.05, 1.3) {$\uafter_V$};

    \draw [domain = -31:31] plot ({cos(\x)}, {sin(\x)});
    \node at (0.7, 0.1) {$\theta$};
  \end{tikzpicture}
  \caption{Projection onto the subspace $V=\vecspan\{v,v'\}$.}
  \label{fig:angles}
\end{figure}

Next, we show that the convex cone of vectors
$v$ and $v'$ is absorbing:

\begin{proposition}
\label{prop:two-advertisers-convergence}
Let $\left\langle v, v' \right\rangle \geq 0$ and take $u$ to be an opinion vector
and $\uafter$ to be a vector resulting from intervening on $u$ with either $v$ or
$v'$. If $u_V$ is a conical combination of $v$ and $v'$, then also $\uafter_V$ is such
a conical combination.
\end{proposition}
\begin{proof}
  Assume wlog that the vector applied is $v$ and
  let $k$ be the same constant as in the proof of
  Proposition~\ref{prop:two-advertisers-subspace}.
  Then,
  \[ \uafter/k =  u + \eta \cdot \left\langle u, v\right\rangle\cdot v = u_V +
    \eta \cdot \left\langle u_V, v\right\rangle\cdot v + u_W\;.\]
  Therefore,
  $\uafter_V$ can be written as a nonnegative linear combination of $u_V$ and $v$, where we use the fact that $\left\langle u_V, v\right\rangle $ is nonnegative, which follows since $u_V$ is a conical combination of $v$ and $v'$, and $\left\langle v, v' \right\rangle \geq 0$.
\end{proof}

Next, we prove that when $\langle v,v' \rangle>0$, the
opinion $u^{(t)}$ not only approaches subspace $V$, but also
a specific area of $V$, namely, either cone$(v,v')$ or cone$(-v,-v')$.

\begin{proposition}
\label{prop:two-advertisers-cvx-hull}
Let $\left\langle v, v'\right\rangle > 0$ and consider a vector $u^{(t)}$
such that $\|u^{(t)}_{V}\|\ge c > 0$. Then, there exists $T := T(c, \theta, \eta)$
such that with probability at least $2^{-T}$,
vector $u_{V}^{(t+T)}$ will either be a conical combination of $v$ and $v'$
\emph{or} a conical combination of $-v$ and $-v'$.
\end{proposition}

\begin{proof}
  First, for any vector $u^{(t)}$ such that $\|u^{(t)}_V\| \ge c > 0$,
  at least one of $v,v',-v,-v'$ has positive inner product with
  $u^{(t)}$ (and $u^{(t)}_V$) which can be lower bounded by
  a function of $c$ and $\theta$ (see Figure~\ref{fig:angles}). Take such a vector and call it $\vbar$.
  By the argument from Proposition~\ref{prop:two-advertisers-subspace}, applying $\vbar$ repeatedly will bring $u^{(t+T)}$ arbitrarily close to it.
  More precisely, for every $\eps>0$, there exists
  $T_1=T_1(c,\theta,\eta,\eps)$ such that
  $\|u_V^{(t+T_1)}-\vbar\|<\eps$
  and $\|u^{(t+T_1)}-\vbar\|<\eps$ both hold.
  
  Furthermore, since $\langle v,v'\rangle>0$,
  there exists $\eps>0$ such that if  
  $\|u^{(t)}-\vbar\|<\eps$, then applying the other 
  intervention vector ($v$ or $v'$) once guarantees that
  $u_V^{(t+1)}$ enters the convex cone between $v$ and $v'$
  or, respectively, between $-v$ and $-v'$.
  In particular, if $u_V^{(t)}$ already is in the convex cone,
  then applying either intervention will keep it inside
  by Proposition~\ref{prop:two-advertisers-convergence}.
  On the other hand, if $u_V^{(t)}$ is not yet in the cone,
  but at the distance at most $\eps$ to $\vbar$, then applying the other 
  intervention will bring it inside the cone
  (see Figure~\ref{fig:angles}).
  
  Therefore, there exists a sequence of $T(c,\theta, \eta)=T_1+1$ interventions that make $u_V^{(t+T)}$ enter cone$(v,v')$ or cone$(-v,-v')$. Clearly, this sequence occurs with probability
  $2^{-T}$.
\end{proof}

We are now ready to prove Theorem~\ref{thm:two-advertisers}.

\begin{proof}[Proof of Theorem~\ref{thm:two-advertisers}]
  Let $\|u^{(1)}_V\| = c > 0$.
  Proposition~\ref{prop:two-advertisers-subspace} tells us that the squared
  norm of the projection $u^{(t)}_W$ onto subspace $W = V^\perp$ never increases,
  and with  probability 1/2 decreases by the multiplicative factor
  $1-\xi(c,\eta,\theta)<1$.
  By induction (note that $\xi$ increases with $c$), $u^{(t)}_W$ converges to 0, and
  consequently $\|u^{(t)}-u^{(t)}_V\|$ converges to 0, almost surely. 

  In order to show that convergence to one of the two convex cones
  occurs, we apply Proposition~\ref{prop:two-advertisers-cvx-hull}.
  Since at \emph{any time step} $t$, there exists a sequence of
  $T$ choices that puts $u^{(t+T)}_V$ in one of the convex cones, and
  since $T$ depends only on the starting parameters $c$, $\theta$, and $\eta$,
  we get that $u^{(t)}_V$ almost surely eventually enters one of the cones.
  By Proposition~\ref{prop:two-advertisers-convergence} and induction,
  once $u^{(t)}_V$ enters a convex cone, it never leaves.
\end{proof}

Proposition~\ref{cor:orthogonal-advertisers} follows
as a corollary of Propositions~\ref{prop:two-advertisers-subspace} and~\ref{prop:two-advertisers-convergence}:

\begin{proof}[Proof of Proposition~\ref{cor:orthogonal-advertisers}]
  The first statement is an inductive application of
  Proposition~\ref{prop:two-advertisers-subspace}, exactly the same as in the
  proof of Theorem~\ref{thm:two-advertisers}.
  
  The second statement follows from noting that out of four orthogonal pairs
  of vectors $\{v, v'\}$, $\{ v, -v'\}$, $\{-v, v'\}$, or $\{-v, -v'\}$, there
  is exactly one such that $u^{(1)}_V$ is a (strict) conical combination of this pair
  (by assuming $\langle u^{(1)},v\rangle\ne 0$ and $\langle u^{(1)},v'\rangle\ne 0$
  we avoid ambiguity in case $u^{(1)}_V$ is parallel to $v$ or $v'$).
  By the same argument as in Proposition~\ref{prop:two-advertisers-convergence}
  and by induction, if the initial projection $u^{(1)}_V$ is strictly inside one
  of the convex cones, the projection
  $u^{(t)}_V$ remains strictly inside forever.
\end{proof}

\subsection{Proof of Theorem~\ref{thm:contraction}}
\label{sec:appendix_pull_strength}


  Consider the subspace $V=\vecspan\{v,v'\}$ with some coordinate system
  (cf.~Figure~\ref{fig:angles}) imposed on it.
  As is standard, a unit vector $u \in V$ can be
  represented in this system by its angle $\alpha(u)\in[0,2\pi)$
  as measured counterclockwise from the positive $x$-axis.

  Given a unit vector $\vbar\in V$, let
  $f_{\vbar}:[0, 2\pi)\to[0,2\pi)$ be the function with the following meaning:
  given a unit vector $u \in V$ with angle $\alpha = \alpha(u)$,
  the value $f_{\vbar}(\alpha) = \alpha(\uafter)$ represents the angle of vector $\uafter$
  resulting from applying intervention $\vbar$
  to vector $u$. Note that $\alpha(\vbar)$ is a fixed point of $f_{\vbar}$.
  Also, by Proposition~\ref{prop:two-advertisers-convergence},
  both functions $f_v$ and $f_{v'}$ map the interval
  corresponding to cone$(v,v')$ to itself.  

  The main part of our argument is the following lemma,
  which we prove last:
  \begin{lemma}\label{lem:fv-contraction}
    If $\langle v,v'\rangle=\cos\theta>1/\sqrt{2+\eta}$, then
    functions $f_v$ and $f_{v'}$ restricted to the convex cone of
    $v$ and $v'$ are contractions, i.e., there exists
    $k=k(\theta, \eta)<1$ such that for all vectors
    $u,u'\in\mathrm{cone}(v,v')$, letting
    $\alpha:=\alpha(u),\beta:=\alpha(u'),\vbar\in\{v,v'\}$, we have
    \begin{align}\label{eq:06}
      \left|f_{\vbar}(\beta)-f_{\vbar}(\alpha)\right| \le k\cdot |\beta-\alpha|\;,
    \end{align}
    where the distances $|f_{\vbar}(\beta)-f_{v^{*}}(\alpha)|$ and $|\beta-\alpha|$
    are in the metric induced by $S^1$, i.e., ``modulo $2\pi$''.
  \end{lemma}

\begin{proof}[Proof of Theorem \ref{thm:contraction}]
  Lemma~\ref{lem:fv-contraction} implies that
  the angle distance between two opinions $u_1^{(t)},u_2^{(t)}\in V$ starting in the convex
  cone deterministically converges to 0 as $t$ goes to infinity.
  Of course, this is equivalent to their Euclidean distance
  $\|u_1^{(t)}-u_2^{(t)}\|$ converging to 0.
  We now make a continuity argument to show that such
  convergence
  almost surely occurs also for general
  $u_1^{(t)}, u_2^{(t)} \in S^{d-1}$. To this end, we let
  $g_v, g_{v'}: S^{d-1}\to[0,2\pi)$ as natural extensions of $f_v,f_{v'}$:
  the value $g_{\vbar}(u)$ denotes the angle of the projection $\uafter_V$ of the new
  opinion onto $V$, after applying $\vbar$ on opinion $u$
  (cf.~Figure~\ref{fig:angles}). Note that the value
  $g_{\vbar}(u)$ depends only on the angle $\alpha(u_V)$ and the orthogonal projection
  length $\|u_W\|$:
  \[
    g_{\vbar}(u) = g_{\vbar}\big(\alpha(u_V), \|u_W\|\big)\;.
  \]
  In this parametrization, for $u\in V$ we have
  $f_{\vbar}(\alpha(u)) = g_{\vbar}(u) = g_{\vbar}(\alpha(u), 0)$.

  By Theorem~\ref{thm:two-advertisers}, for any starting
  opinions $u_1^{(1)}$ and $u_2^{(1)}$ having non-zero
  projections onto $V$,
  almost surely there exists a $t$ such that $(u_1^{(t)})_V$
  and $(u_2^{(t)})_V$ end up inside (possibly different) convex cones. We consider the
  case of $u_1^{(t)}$ and $u_2^{(t)}$ both in cone$(v, v')$, other three cases being
  analogous. Furthermore, almost surely, $\|(u_1^{(t)})_W\|$ and $\|(u_2^{(t)})_W\|$ converge
  to 0. Hence, it is enough that we show that almost surely
  $|\alpha((u_1^{(t)})_V)-\alpha((u_2^{(t)})_V)|$ (in $S^1$ distance)
  converges to zero.

  To this end, let $\delta > 0$. By uniform continuity of $g_v$, we know that
  for small enough value of $r$, we have
  \[
    \left| g_v(\alpha, r) - g_v(\alpha, 0)\right| < \frac{1-k}{4}\cdot \delta
  \]
  for every $\alpha\in[0,2\pi)$, where $k$ is the Lipschitz constant
  from~\eqref{eq:06}. Therefore, almost surely, for $t$ large enough,
  for $u_1^{(t)}$ and $u_2^{(t)}$ parameterized as 
  $u_1^{(t)}=(\alpha_1, r_1)$ and
  $u_2^{(t)}=(\alpha_2, r_2)$ we have
  \begin{align*}
    \left|g_v(\alpha_1, r_1)-g_v(\alpha_2,r_2)\right|
    &\le
      \left|g_v(\alpha_1,r_1)-g_v(\alpha_1,0)\right|+
      \left|g_v(\alpha_1,0)-g_v(\alpha_2,0)\right|+
      \left|g_v(\alpha_2,0)-g_v(\alpha_2,r_2)\right|\\
    &\le\frac{1-k}{4}\cdot\delta+
      k\cdot|\alpha_1-\alpha_2|+
      \frac{1-k}{4}\cdot\delta
    \le \left(k + \frac{1-k}{2}\right)\cdot\max(|\alpha_1-\alpha_2|,\delta)\;.
  \end{align*}
  Since $k+(1-k)/2 < 1$, and applying
  the same argument to $f_{v'}$, we conclude by induction that
  the distance $|\alpha_1(t)-\alpha_2(t)|$ must
  decrease and stay below $\delta$ in a finite number of steps.
  Since $\delta > 0$ was arbitrary, it must be
  that $|\alpha_1(t)-\alpha_2(t)|$ converges to 0, 
  concluding the proof of Theorem~\ref{thm:contraction}.\hfill\qedsymbol

It remains to prove Lemma~\ref{lem:fv-contraction}:
\proof{Proof of Lemma~\ref{lem:fv-contraction}.}
  Recall that we assumed a two-dimensional coordinate system on $V$.
  Let $f:=f_{(1,0)}$, i.e., $f$ corresponds to the intervention along the
  $x$-axis in this coordinate system. Clearly, functions $f_v$ and $f_{v'}$ are cyclic shifts of $f$ modulo $2\pi$.
  More precisely, we have
  \begin{align}\label{eq:05}
    f_{\vbar}(\alpha) = \alpha(\vbar) + f\big(\alpha-\alpha(\vbar)\big)\;,
  \end{align}
  where arithmetic in~\eqref{eq:05} is modulo $2\pi$. Furthermore,
  $f$ is symmetric around the intervention vector, i.e.,
  $f(\alpha)=2\pi-f(2\pi-\alpha)$ for $0<\alpha\le\pi$.
  Hence, to prove that $f_v$ and $f_{v'}$ restricted to cone$(v,v')$
  are contractions, it is enough that we show that $f$
  restricted to the interval $[0, \theta]$ is a contraction
  (recall that we assumed $\cos^2(\theta)>1/(2+\eta)$).
  
  \begin{figure}[ht]\centering
    \includegraphics[width=0.6\textwidth, trim={0 4.5cm 0 4cm}, clip]
    {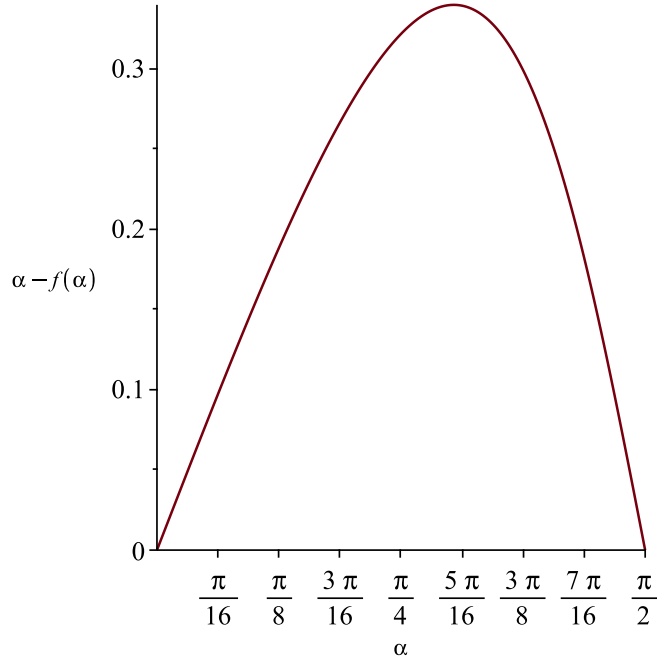}
    \caption{The graph of the ``pull function'' $\alpha-f(\alpha)$ in
      case $\eta=1$.}
    \label{fig:pull}
  \end{figure}
  
  To that end, we use~\eqref{eq:main} to calculate the formula for $f$ for
  $0\le\alpha\le\pi/2$ as
\begin{equation}
\label{eq:pull_function}
   f(\alpha)=\arccos\left(
   \frac{(1+\eta)\cos\alpha}{\sqrt{1+(2\eta+\eta^2)\cos^2\alpha}}\right)\;.
\end{equation}
    
  More computation using elementary calculus
  (we omit the details) establishes that, additionally, for every
  $0\le\alpha<\beta\le\pi/2$:
  \begin{enumerate}
  \item $f(\alpha) \le \alpha$. In other words, applying the intervention
    brings vector $u$ closer to the intervention vector.
  \item $f(\alpha) < f(\beta)$, i.e., applying the intervention does not
    change relative ordering of vectors wrt the intervention vector.
  \item If $\beta\le\theta^*:=\arccos\left(\sqrt{\frac{1}{2+\eta}}\right)$, then
    $0\le \alpha-f(\alpha) < \beta-f(\beta)$, i.e., in absolute terms,
    the ``pull'' on a vector is stronger the further away it
    is from the intervention vector (until the correlation reaches the threshold
    $1/\sqrt{2+\eta}$, cf.~Figure~\ref{fig:pull}).
  \end{enumerate}
  The preceding items taken together imply that for every $0\le\alpha<\beta\le\theta^*$
  we have $0 < f(\beta)-f(\alpha) < \beta-\alpha$. To conclude that $f$ is
  a contraction, we observe that $f$ and its derivative $f'$ are continuous on the interval $[0,\theta^*]$.
  If there exist sequences $(\alpha_k)$ and $(\beta_k)$
  in $[0,\theta]$ for $\theta<\theta^*$
  such that
  $|f(\alpha_k)-f(\beta_k)|/|\beta_k-\alpha_k|$ converges to 1,
  then, by compactness, there exist convergent sequences 
  $\alpha_k\to \alpha^*$ and $\beta_k\to\beta^*$ such that
  $|f(\alpha_k)-f(\beta_k)|/|\beta_k-\alpha_k|\to 1$. Then,
  \begin{enumerate}
      \item Either $\alpha^*\ne\beta^*$ and by continuity we get
      $f(\beta^*)-f(\alpha^*)=\beta^*-\alpha^*$, contradicting
      the third property above.
      \item Or $\alpha^*=\beta^*$, which by continuity of $f'$ implies
      $f'(\alpha^*)=1$ for some $0\le\alpha^*<\theta^*$. But
      that would imply that the derivative of $\alpha-f(\alpha)$,
      i.e., $1-f'(\alpha)$, vanishes at $\alpha^*<\theta^*$,
      again contradicting the third property above
      (see also Figure~\ref{fig:pull}).\qedhere
  \end{enumerate}
\end{proof}


\appendix

\section{Proof of Claim~\ref{cl:sin-calculation}}\label{app:sin-proof}
  Let us embed our underlying space $\mathbb{R}^2$ in $\mathbb{R}^3$
  by setting the last coordinate to zero.
  Letting $\times$ denote the cross product, we have
  \begin{align*}
    u\times u'=(0,0,\sin\alpha_t)\;,
    \qquad\qquad
    f(u,v)\times f(u',v)=(0,0,\sin\alphawhat)\;.
  \end{align*}
  Since the case $\alpha_t\in\{0,\pi\}$ is easily handled
  by noticing that $\alphawhat=\alpha_t$,
  we can assume that
  $0<\alpha_t<\pi$. In that case, 
  it is enough that we prove
  \begin{align}\label{eq:24}
      \left\langle
      u\times u', f(u,v)\times f(u',v)
      \right\rangle=\sin\alpha_t\sin\alphawhat \ge 0\;.
  \end{align}
  Setting 
  $C(w):=\sqrt{1+(2\eta+\eta^2)\langle w,v\rangle^2}$,
  we apply~\eqref{eq:main} and bilinearity of cross product to
  compute
  \begin{align}
      f(u,v)\times f(u',v)
      &= \frac{1}{C(u)C(u')}\Big(
      u\times u' + \eta\big(
      \langle u,v\rangle(v\times u')
      +\langle u',v\rangle(u\times v)
      \big)\Big)\nonumber\\
      &=\frac{1}{C(u)C(u')}\Big(
      u\times u' + \eta\big(
      u\times u'+(\langle u,v\rangle v-u)\times
      (u'-\langle u',v\rangle v)
      \big)\Big)\label{eq:22}\\
      &=\frac{1+\eta}{C(u)C(u')}\cdot u\times u' \;,\label{eq:23}
  \end{align}
  where in~\eqref{eq:22} we used the identity
  $a\times b+c\times d=a\times d+c\times b+(a-c)\times (b-d)$,
  and in~\eqref{eq:23} we used that both
  $\langle u,v\rangle v-u$ and $u'-\langle u',v\rangle v$
  are projections of vectors onto the line orthogonal to $v$,
  and therefore they are parallel and their cross product vanishes.
  
  Consequently, we conclude that 
  $f(u,v)\times f(u',v)$ is parallel to $u\times u'$ with
  a positive proportionality constant, which implies~\eqref{eq:24} and concludes the proof.
  \hfill\qedsymbol

\section{Example with two advertisers}\label{sec:example-two}

For another slightly more involved example, suppose there are two advertisers marketing their products. Agents' opinions now have
five dimensions ($d=5$) with the fourth and fifth coordinates corresponding to the opinions on these two products.
Initially, $500$ opinions on the first three coordinates are distributed randomly and uniformly on a
three-dimensional sphere, and the last two coordinates are equal to zero:
\begin{align*}
    u_i = (u_{i,1},u_{i,2},u_{i,3},0,0)
    \qquad\text{ subject to }\qquad
    u_{i,1}^2+u_{i,2}^2+u_{i,3}^2=1\;.
\end{align*}
Suppose the two advertisers apply interventions $v_1$ and $v_2$
in an alternating fashion. We take $v_1$ and $v_2$ to be orthogonal,
letting
\begin{align*}
    v_1=(\beta, 0, 0, \alpha, 0)\;,\qquad\qquad
    v_2=(0, \beta, 0, 0, \alpha)\;,\qquad\qquad
    \alpha=\frac{3}{4},\beta=\sqrt{1-\alpha^2}\;.
\end{align*}
We proceed to apply $v_1$ and $v_2$ in an alternating fashion.
In Figure~\ref{fig:example-two} we illustrate the agents' opinions after each
advertiser applied their intervention two, four and six times 
(so the total of, respectively, four, eight and twelve interventions
have been applied). A pattern of polarization on the fourth and fifth
coordinates can be observed. At the same time, the pattern on the first
three coordinates is more complicated. The opinions on these dimensions are scattered around
a circle on the plane spanned by the first two coordinates.
This is a somewhat special behavior that arises because vectors
$v_1$ and $v_2$ are orthogonal. 
It is connected to the difference between
Theorem~\ref{thm:two-advertisers} and
Proposition~\ref{cor:orthogonal-advertisers}
discussed in Section~\ref{sec:two-advertisers-intro}.

\begin{figure}[!htp]
\begin{tabular}{c}
  $t=5$\putindeepbox[7pt]
  {\includegraphics[width=0.7\textwidth,height=2.4in]{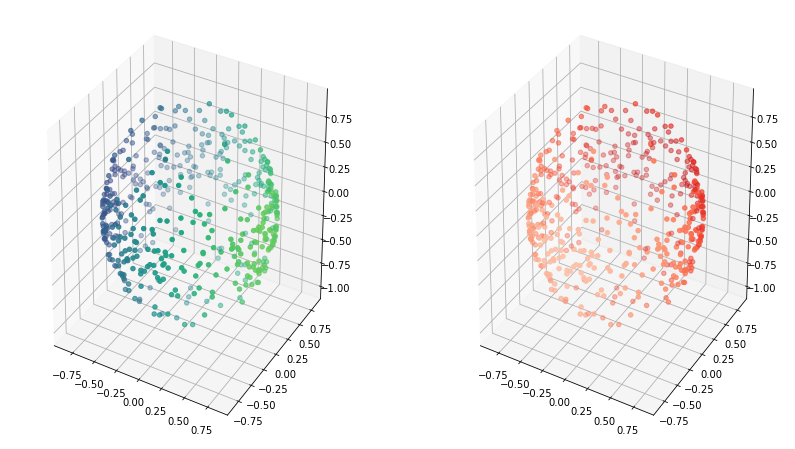}}\\
  $t=9$\putindeepbox[7pt]
  {\includegraphics[width=0.7\textwidth, height=2.4in]{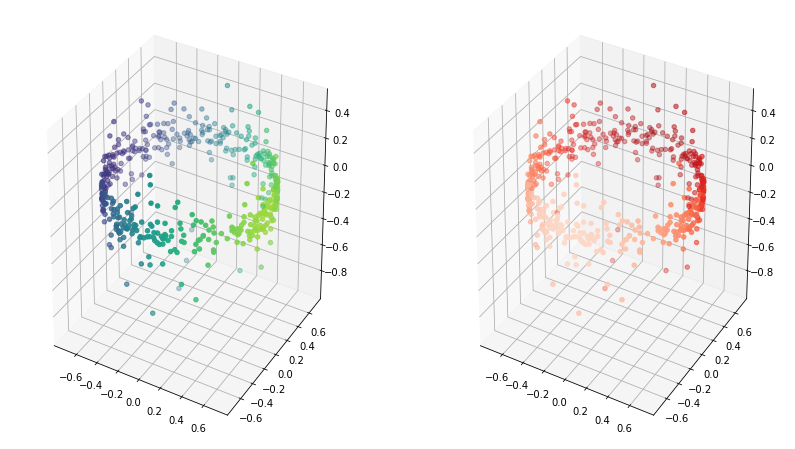}}\\
  $t=13$\putindeepbox[7pt]
  {\includegraphics[width=0.7\textwidth,height=2.4in]{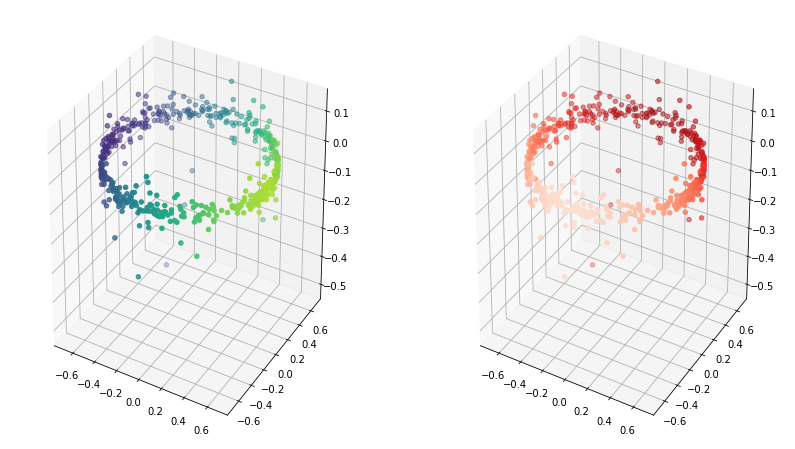}}\\
  \begin{tabular}{cc}
    {\includegraphics[width=0.45\textwidth]{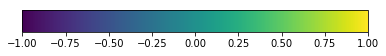}}&
    {\includegraphics[width=0.45\textwidth]{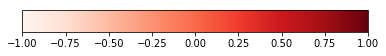}}
  \end{tabular}
\end{tabular}
\caption{Illustration of the process described in Appendix~\ref{sec:example-two}.
This time we need to visualize five dimensions. This is done
with spatial positions for the first three dimensions
$j=1,2,3$ and
two different color scales for $j=4,5$. Accordingly, 
two figures are displayed for each time step $t=5,9,13$.
In each pair of
figures the points in the left figure have the same spatial 
positions as in the right figure and the colors illustrate dimensions $j=4$ (on
the left) and $j=5$ (on the right). 
}
\label{fig:example-two}
\end{figure}

\section{Proof of Proposition~\ref{prop:polarization-cost}}
\label{app:polarization-cost}

Recall that the two-agent intervention maximizes $\min(\uafter_{1,d},\uafter_{2,d})$.
Due to symmetry, we will consider wlog the one-agent intervention that maximizes
$\uafter_{1,d}$. Substituting into~\eqref{eq:main}, we get that applying an intervention
$v$ results in
\begin{align}\label{eq:19}
    \uafter_{i,d} =\frac{\langle u_i,v\rangle\cdot v_d}{\sqrt{1+3\langle u_i,v\rangle^2}}\;. 
\end{align}
Recalling~\eqref{eq:09}, we can apply any unitary transformation on the opinions without changing the correlations, and hence assume that
\begin{align}\label{eq:16}
    u_1:= (\sin\alpha, \cos\alpha, 0, \ldots, 0)\;,\qquad\qquad
    u_2:= (-\sin\alpha, \cos\alpha, 0,\ldots, 0)\;
\end{align}
for some $0\le\alpha\le\pi/2$ and accordingly,
$c=\cos^2\alpha-\sin^2\alpha = \cos(2\alpha)$. In particular,
$\alpha=0$ means that the agents are in full agreement,
$\alpha=\pi/4$ corresponds to the case of orthogonal opinions
and $\alpha=\pi/2$ is the case where the opinions are antipodal.

Assuming~\eqref{eq:16},
once we fix the first
two coordinates of the intervention $v_1$ and $v_2$,
also the values of $\langle u_1,v\rangle$
and $\langle u_2,v\rangle$ become fixed.
Therefore, due to~\eqref{eq:19},
the values of $\uafter_{i,d}$ depend only on $v_d$ in a linear fashion. Accordingly, the influencer should place as
much weight as possible on the last coordinate and we can
conclude that both two- and one-agent interventions have $v_{j}=0$ for $2<j<d$.
Hence, in the following we will assume wlog that $d=3$, 
$u_1=(\sin\alpha,\cos\alpha,0)$ and $u_2=(-\sin\alpha,\cos\alpha,0)$
(see Figure~\ref{fig:interventions}).

\begin{figure}[!ht]\centering\begin{tikzpicture}
    \draw [->] (0, -1) -- (0, 3);
    \draw [->] (-3, 0) -- (3, 0);
    \draw [->] (8, -1) -- (8, 3);
    \draw [->] (5, 0) -- (11, 0);

    \draw [->, thick] (0, 0) -- (2.4, 1.8);
    \draw [->, thick] (0, 0) -- (-2.4, 1.8);
    \draw [->, thick, color=red] (0, 0) -- (1.39, 1.04);

    \node at (2.8, 1.8) {$u_1$};
    \node at (-2.8, 1.8) {$u_2$};
    \node at (2.8, 1.04) {$v^{\perp}=\sqrt{3}/3\cdot u_1$};

    \draw [domain=37:90] plot ({cos(\x)}, {sin(\x)});
    \draw [domain=90:143] plot ({1.2*cos(\x)}, {1.2*sin(\x)});
    \node at (0.3, 0.6) {$\alpha$};
    \node at (-0.3, 0.7) {$\alpha$};
    
    \draw [->, thick] (8, 0) -- (10.4, 1.8);
    \draw [->, thick] (8, 0) -- (5.6, 1.8);
    \draw [->, thick, color=red] (8, 0) -- (8, 1.71);

    \node at (10.8, 1.8) {$u_1$};
    \node at (5.2, 1.8) {$u_2$};
    \node at (8.5, 1.71) {$v^\perp$};

    \draw [domain=37:90] plot ({8+cos(\x)}, {sin(\x)});
    \draw [domain=90:143] plot ({8+1.2*cos(\x)}, {1.2*sin(\x)});
    \node at (8.3, 0.6) {$\alpha$};
    \node at (7.7, 0.7) {$\alpha$};    
  \end{tikzpicture}
  \caption{The projection of one-agent (left) and two-agent (right)
  interventions onto the first two dimensions.}
  \label{fig:interventions}
\end{figure}
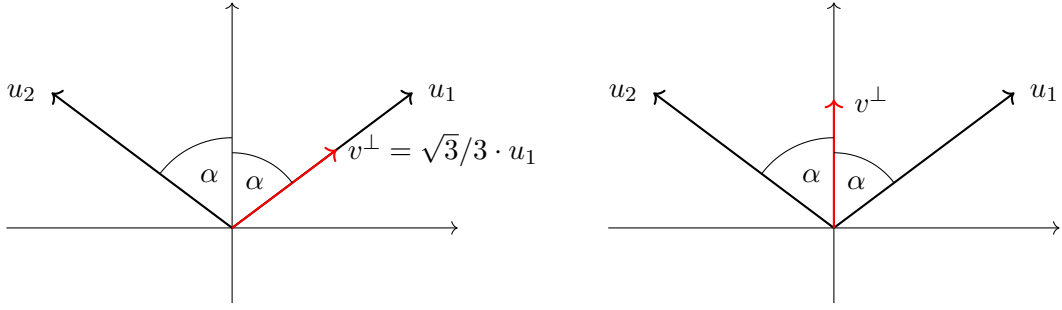

First, consider the one-agent intervention maximizing $\uafter_{1,3}$.
Clearly, the intervention should be of the form
\begin{align*}
    \vone = \cos\beta \cdot u_1 + \sin\beta \cdot (0, 0, 1)
\end{align*}
for some $0\le\beta\le\pi/2$. Substituting in~\eqref{eq:19}, we compute
\begin{align}\label{eq:01}
    (\uafter_{1,3})^2 = \frac{\cos^2\beta\sin^2\beta}{1+3\cos^2\beta} \;.
\end{align}
Maximizing \eqref{eq:01}, we get the maximum at
$\cos\beta = \sqrt{3}/3$ and 
\begin{align*}
    \vone=\frac{\sqrt{3}}{3}\cdot u_1+\frac{\sqrt{6}}{3}\cdot (0,0,1)\;,
\end{align*}
resulting in $\uafter_{1,3} = 1/3$. The value $1/3$ is
the benchmark for what can be achieved by one intervention.
It is a maximum value for $\uafter_{1,3}$ attainable provided that initially
$u_{1,3}=0$.

What is the effect of this intervention on the other opinion
$u_2$? Since $\langle u_2,\vone\rangle=\sqrt{3}c/3$,
substituting into \eqref{eq:19} we get
\begin{align*}
  \uafter_{2,3} = \frac{\sqrt{3}c/3\cdot \sqrt{6}/3}{\sqrt{1+c^2}}=
    \frac{c\sqrt{2}}{3\sqrt{1+c^2}}
    \; .
\end{align*}
The value of $\uafter_{2,3}$ as a function of the correlation $c\in[-1,1]$ is shown in
blue in Figure~\ref{fig:tradeoff-example-right}. In particular, it increases
from $-1/3$ to $1/3$, passing through $0$ for $c=0$.

Moving to the two-agent intervention, in this case it is not difficult
to see (cf.~Figure~\ref{fig:interventions})
that the intervention vector should be of the form
\begin{align*}
  \vtwo = (0, \cos\beta, \sin\beta)
\end{align*}
for some $0\le\beta\le\pi/2$. A computation in
a computer algebra system (CAS) establishes
that $\uafter_{1,3}=\uafter_{2,3}$ is maximized for
\begin{align*}
  \cos^2\beta =
  \frac{\sqrt{2}(\sqrt{3c+5}-\sqrt{2})}{3(c+1)}\;,
\end{align*}
yielding an expression
\begin{align*}
    \uafter_{1,3} = \uafter_{2,3} =
  \sqrt{\frac{3c+7-2\sqrt{6c+10}}{9(c+1)}}\;.
\end{align*}
This function is depicted in Figure~\ref{fig:tradeoff-example-right}
in red. In particular, for $c\in[-1, 1]$, it increases from
$0$ to $1/3$ and its value at $c=0$ is approximately $0.27$.
Furthermore, its growth close to $c=-1$ is of the square-root type.

Turning to the new correlation values $\cone$ and $\ctwo$, another CAS computation using the formulas for $\vone$ and $\vtwo$
gives
\begin{align*}
  \cone 
  =
    \frac{c\sqrt{2}}{\sqrt{c^2+1}}\;,\qquad\qquad
    \ctwo
    =
      1-\frac{\sqrt{2}(1-c)}{\sqrt{3c+5}}\;,
\end{align*}
establishing~\eqref{eq:25}. To conclude the proof we need another elementary calculation showing that 
$\ctwo\ge \cone$ always holds. We omit the details, referring to Figure~\ref{fig:tradeoff-example-left}
and noting that in the critical region for $c=1-\eps$ we have
\begin{flalign*}
    \qquad\qquad\qquad\qquad\qquad
    \ctwo=1-\frac{1}{2}\eps-\frac{3}{32}\eps^2+O(\eps^3)
    \ge
    \cone=1-\frac{1}{2}\eps-\frac{3}{8}\eps^2+O(\eps^3)\;.
    &&\qedsymbol
\end{flalign*}

\section{Proof of Proposition~\ref{prop:short-term-cap}}
\label{app:short-term-cap}

Let us write a generic intervention vector as
\begin{align*}
  v= (\cos\beta \cdot v^*, \sin \beta)\;, 
\end{align*}
where $0\le\beta\le\pi/2, v^*\in\mathbb{R}^{d-1}$ and $\|v^*\|=1$. If
$v$ is applied to an opinion vector $u_i = (u_i^*, 0)$ and we let
$c_i := \langle u_i^*,v^*\rangle$, substituting into~\eqref{eq:main}
we can compute
\begin{align*}
  u_i + \langle u_i,v\rangle\cdot v=
  (u^*_i,0)+c_i\cos\beta(\cos\beta\cdot v^*,\sin\beta)=
  (u_i^*+c_i\cos^2\beta\cdot v^*,c_i\cos\beta\sin\beta)\;,
\end{align*}
and therefore, using~\eqref{eq:main3},
\begin{align}\label{eq:02}
  \uafter_{i,d} = \frac{c_i\cos\beta\sin\beta}
  {\sqrt{1+3c_i^2\cos^2\beta}}
  = \frac{c_iz\sqrt{1-z^2}}{\sqrt{1+3c_i^2z^2}}\;,
\end{align}
where we let $z:=\cos\beta$.

Consider a fixed unit vector $v^*\in\mathbb{R}^{d-1}$. In order to maximize
$\uafter_{i,d}$ for an opinion $u_i$ with $\langle u_i^*, v^*\rangle=c_i$,
we need to optimize over $z$ in~\eqref{eq:02},
resulting in $z = \sqrt{\sqrt{1+3c_i^2}-1}/(\sqrt{3}c_i)$
and, substituting,
\begin{align}\label{eq:03}
  \uafter_{i,d} = \frac{\sqrt{1+3c_i^2}-1}{3c_i}\;.
\end{align}
The right-hand side
of~\eqref{eq:02} is easily seen to be increasing 
in $c_i>0$ for a fixed $z$. Therefore,
in order to maximize the number of points with $\uafter_{i,d}> T$ for a fixed
$v^*$, we solve the equation $T = \frac{\sqrt{1+3c^2}-1}{3c}$
for $c$, resulting in $c=\frac{2T}{1-3T^2}$ and apply
the intervention
\[
    v = (\cos\beta\cdot v^*,\sin\beta)\;,
\]
just as claimed in~\eqref{eq:08}. This intervention results in 
$\uafter_{i,d}>T$
for all opinions satisfying $\langle u_i^*,v^*\rangle>c$.
In other words, the objective $\uafter_{i,d}> T$ is achieved
for exactly those opinions contained in the spherical cap
$\{x\in\mathbb{R}^{d-1}: \langle x,v^*\rangle > c\}$.
Maximizing over all directions $v^*\in\mathbb{R}^{d-1}$ completes the
proof.\hfill\qedsymbol

\bibliographystyle{alpha} 
\bibliography{bibliography_arxiv_v4}

\end{document}